\documentclass[final,3p,authoryear,times]{elsarticle} 

\usepackage{amsthm}
\usepackage{amssymb}
\usepackage{mathtools} 
\usepackage{rotating}
\usepackage[latin1]{inputenc}
\usepackage{graphicx}
\usepackage{url}
\usepackage{amsmath}
\usepackage{amsfonts}
\usepackage{url}
\usepackage{color}
\usepackage[table]{xcolor}
\usepackage{subfigure}
\usepackage{slashbox}
\usepackage{ctable}
\usepackage{floatrow}
\usepackage{tikz}
\usepackage{xcolor}
\usepackage{todonotes} 
\DeclareMathOperator{\tr}{\text{tr}}
\newfloatcommand{capbtabbox}{table}[][\FBwidth]

\usepackage{blindtext}

\newcommand{\real}{\text{\rm I\!R}}

\def\bx{\boldsymbol{x}}
\def\by{\boldsymbol{y}}
\def\bw{\boldsymbol{w}}
\def\bu{\boldsymbol{u}}
\def\bv{\boldsymbol{v}}
\def\bz{\boldsymbol{z}}
\def\bX{\boldsymbol{X}}
\def\bY{\boldsymbol{Y}}
\def\bW{\boldsymbol{W}}

\def\balpha{\boldsymbol{\alpha}}
\def\bbeta{\boldsymbol{\beta}}
\def\bmu{\boldsymbol{\mu}}
\def\bpi{\boldsymbol{\pi}}
\def\bSigma{\boldsymbol{\Sigma}}

\def\bvartheta{\boldsymbol{\vartheta}}

\newtheorem{pro}{Proposition} 
\renewenvironment{proof}{\textit{Proof.}}{\qed}

\journal{arXiv.org}

\begin{document}

\begin{frontmatter}

\title{
Robust Clustering in Regression Analysis via the\\ Contaminated Gaussian Cluster-Weighted Model
}

\author[CT]{Antonio Punzo\corref{cor}} 
\ead{antonio.punzo@unict.it}
\author[CA]{Paul D.~McNicholas} 
\ead{mcnicholas@math.mcmaster.ca}
\cortext[cor]{Corresponding author: 
Email: \texttt{antonio.punzo@unict.it}, 
Phone: +39-095-7537640, 
Fax: +39-095-7537610}
\address[CT]{
Department of Economics and Business, University of Catania, Catania, Italy.
							}
\address[CA]{
Department of Mathematics \& Statistics, McMaster University, Hamilton, Canada.
}

\begin{abstract}
\small
The Gaussian cluster-weighted model (CWM) is a mixture of regression models with random covariates that allows for flexible clustering of a random vector composed of response variables and covariates.
In each mixture component, it adopts a Gaussian distribution for both the covariates and the responses given the covariates.
To robustify the approach with respect to possible elliptical heavy tailed departures from normality, due to the presence of atypical observations, the contaminated Gaussian CWM is here introduced.
In addition to the parameters of the Gaussian CWM, each mixture component of our contaminated CWM has a parameter controlling the proportion of outliers, one controlling the proportion of leverage points, one specifying the degree of contamination with respect to the response variables, and one specifying the degree of contamination with respect to the covariates. 
Crucially, these parameters do not have to be specified \textit{a~priori}, adding flexibility to our approach.
Furthermore, once the model is estimated and the observations are assigned to the groups, a finer intra-group classification in typical points, outliers, good leverage points, and bad leverage points --- concepts of primary importance in robust regression analysis --- can be directly obtained. 
Relations with other mixture-based contaminated models are analyzed, identifiability conditions are provided, an expectation-conditional maximization algorithm is outlined for parameter estimation, and various implementation and operational issues are discussed.  
Properties of the estimators of the regression coefficients are evaluated through Monte Carlo experiments and compared to the estimators from the Gaussian CWM. 
A sensitivity study is also conducted based on a real data set.
\end{abstract}

\end{frontmatter}


\section{Introduction}
\label{sec:Introduction}

Given a continuous $d_{\bW}$-variate random variable $\bW$, with density $p\left(\bw\right)$, finite mixtures of (continuous) distributions constitute both a flexible way for density estimation and a powerful device for clustering and classification by often assuming that each mixture component represents a group (or cluster or class) in the original data (see, e.g., \citealp{Titt:Smit:Mako:stat:1985}, \citealp{McLa:Basf:mixt:1988}, and \citealp{McLa:Peel:fini:2000}). 


In many applied problems, the variable of interest $\bW$ is composed by a $d_{\bY}$-variate response variable $\bY$ and by a random covariate $\bX$ of dimension $d_{\bX}$, with $d_{\bX}+d_{\bY}=d_{\bW}$; that is, $\bW=\left(\bX,\bY\right)$.
In such a case mixtures of distributions, that fail to incorporate a possible local (i.e., within-group) relation of $\bY$ on $\bX=\bx$, may perform inadequately.
A valid alternative, in the mixture modeling framework, is represented by mixtures of regression models (see \citealp{DeSa:Cron:Amax:1988} and \citealp[][Chapter~8]{Fruh:Fine:2006} for details).
In turn, this family of models can be split into two sub-families: mixtures of regression models with fixed covariates and mixtures of regression models with random covariates.
However, as stated by \citet{Henn:Iden:2000}, the former subfamily is inadequate for most of the applications because it assumes ``assignment independence'', i.e., that the probability for a point $\left(\bx,\by\right)$ to be generated by one of the groups distributions has to be the same for all covariate values~$\bx$. In other words, the assignment of the data points to the groups has to be independent of the covariates.
On the contrary, mixtures of regression models with random covariates --- which this paper focuses on --- assume ``assignment dependence'' by allowing the assignment of the data points to the groups to depend on $\bX$.

An eminent member in the class of mixtures of regression models with random covariates is represented by the cluster-weighted model \citep[CWM;][]{Gers:Nonl:1997}, also called the saturated mixture regression model \citep{Wede:Conc:2002}.
The CWM factorizes $p\left(\bx,\by\right)$, in each mixture component, into the product between the conditional density of $\bY|\bx$ and the marginal density of $\boldsymbol{X}$ by assuming a local regression of $\bY$ on $\bx$. 
The distribution of $\bX$ can differ across groups and this allows for assignment dependence.
Recent work on cluster-weighted modeling can be found in \citet{Ingr:Mino:Vitt:Loca:2012}, \citet{Sube:Punz:Ingr:McNi:Clus:2013}, \citet{Ingr:Mino:Punz:Mode:2014}, and \citet{Punz:Flex:2014}.
With respect to our framework, characterized by a possible multivariate response variable, \citet{Dang:Punz:McNi:Ingr:Brow:Mult:2014} propose the Gaussian CWM with density
\begin{equation}
p\left(\bx,\by;\bvartheta\right)=\sum_{j=1}^k\pi_j\phi\left(\by;\bmu_{\bY}\left(\bx;\bbeta_j\right),\bSigma_{\bY|j}\right)\phi\left(\bx;\bmu_{\bX|j},\bSigma_{\bX|j}\right),
\label{eq:Gaussian CWM}
\end{equation}
where $\pi_j$ are positive weights summing to one, $\phi\left(\cdot;\bmu,\bSigma\right)$ denotes the density of a Gaussian random vector with mean $\bmu$ and covariance matrix $\bSigma$, and $\bmu_{\bY}\left(\bx;\bbeta_j\right)=E\left(\bY|\bx,j\right)=\bbeta_j'\bx^*$ denotes the local conditional mean of $\bY|\bx$, with $\bbeta_j$ being a vector of regression coefficients of dimension $\left[\left(1+d_{\bX}\right) \times d_{\bY}\right]$ and $\bx^*=\left(1,\bx\right)$ to account for the intercept(s).
In \eqref{eq:Gaussian CWM}, $\bvartheta$ contains all of the parameters of the model.

Unfortunately, real data are often ``contaminated'' by atypical observations that affect the estimation of the model parameters with particular interest, in the regression context, to the regression coefficients.
Accordingly, the detection of these atypical observations, and the development of robust methods of parameters estimation insensitive to their presence, is an important problem.
However, as emphasized by \cite{Davi:Gath:Thei:1993} --- see also \citealp{Henn:Fixe:2002} --- atypical observations should be defined with respect to a reference distribution.
That is, the shape (i.e., distribution) of the typical points has to be assumed in order to define what an atypical point is; in this way, the region of atypical points can be defined, e.g., as a region where the density of the reference distribution is low. 
If the reference distribution is chosen to be Gaussian, as for example in model~\eqref{eq:Gaussian CWM}, a common and simple elliptical generalization, having heavier tails for the occurrence of atypical points, is the contaminated Gaussian distribution; it is a two-component Gaussian mixture in which one of the components, with a large prior probability, represents the typical observations (reference distribution), and the other, with a small prior probability, the same mean, and an inflated covariance matrix, represents the atypical observations \citep{Aitk:Wils:Mixt:1980}.

Based on the above considerations, this paper introduces the contaminated Gaussian CWM, obtained from \eqref{eq:Gaussian CWM} by substituting the Gaussian distribution with the contaminated Gaussian distribution. 
Interestingly, each component joint density of the proposed model adheres to the taxonomy of atypical observations which is commonly considered in regression analysis; such a taxonomy will be recalled here \citep[for further details see, e.g.,][Chapter~1]{Rous:Lero:Robu:2005}.
In regression analysis, atypical observations can be distinguished between two types. 
Atypical observations in $\bY|\bx$ represent model failure. 
Such observations are called (vertical) outliers.
Atypical observations with respect to $\bX$ are called leverage points. 
In regression it helps to make a distinction between two types of leverage points: good and bad. 
A bad leverage point is a regression outlier that has an $\bx$ value that is atypical among the values of $\bX$ as well. 
A good leverage point is a point that is unusually large or small among the $\bX$ values but is not a regression outlier ($\bx$ is atypical but the corresponding $\by$ fits the model quite well).
A point like this is called ``good'' because it improves the precision of the regression coefficients \citep[][p.~635]{Rous:VanZ:Unma:1990}.
Each point $\left(\bx,\by\right)$ can be so labeled in one of the four categories indicated in \tablename~\ref{tab:Atypical observation labelling}.
\begin{table}[!ht]
\caption{
\label{tab:Atypical observation labelling}
Atypical observation labelling.
}
\centering
\begin{tabular}{ccc}
\toprule
\backslashbox{Outlier (on $\bY|\bx$)}{Leverage (on $\bX$)} 	 &	Yes  &	No	 \\
\midrule															
Yes   & bad leverage  &     outlier      \\
No    & good leverage & typical (bulk of the data)  \\
\bottomrule	
\end{tabular}
\end{table} 

As it will be better explained in Section~\ref{subsec:Automatic detection of noise}, once the contaminated Gaussian CWM is fitted to the observed data, by means of maximum \textit{a~posteriori} probabilities, each observation can be first assigned to one of the $k$ groups and then classified into one of the four categories defined in \tablename~\ref{tab:Atypical observation labelling}; thus, we have a model for simultaneous clustering and detection of atypical observations in a regression context.

In the mixtures of regression models framework, other solutions for robust clustering exist.
Some recent proposals are given in the following:
\begin{enumerate}
	\item \citet{Gali:Soff:Amul:2014} propose a mixture of parallel regression models with $t$-distributed errors;
	\item \citet{Yao:Wei:Yu:Robu:2014} introduce mixtures of regression models with $t$-distributed errors;
	\item \citet{Song:Yao:Xing:Robu:2014} define mixtures of regression models with Laplace-distributed errors;
	\item \citet{Ingr:Mino:Vitt:Loca:2012} propose the $t$ CWM, where the Gaussian distribution in \eqref{eq:Gaussian CWM} is replaced by a $t$ distribution \citep[see also][]{Ingr:Mino:Punz:Mode:2014}.
\end{enumerate}
In general, with respect to our approach, these four models have some drawbacks.
First, they do not allow for the direct detection of atypical observations.
Actually, for the $t$-based models, a procedure described by \citet[][p.~232]{McLa:Peel:fini:2000} could be eventually adopted to classify the observations as atypical.
The procedure stems from a $\chi^2$-approximation of the squared Mahalanobis distance of each observation after its maximum \textit{a~posteriori} classification to one of the $k$ groups.
However, the procedure is not direct and it is not corroborated by the theory. 
Second, the first three models do not consider the presence of possible leverage points in each group; moreover, they belong to the class of mixture of regression models with fixed covariates and, as such, assume assignment independence.
The first model is also based on the assumption of parallel local regression models.
It is also noteworthy that only the first model considers a possible multivariate response variable $\bY$.
For further mixture-based approaches for robust clustering in regression analysis, see, e.g., \citet{Neyk:Filz:Dimo:Neyt:Robu:2007} and \citet{Bai:Yao:Boyer:Robu:2012}.   

The paper is organized as follows.
The contaminated Gaussian CWM is presented in Section~\ref{sec:The model} and compared to other mixture-based contaminated approaches in Section~\ref{sec:Relation with other contaminated models}.
Sufficient conditions for identifiability are given in Section~\ref{sec:Identifiability}, and an expectation-conditional maximization (ECM) algorithm for maximum likelihood parameter estimation is outlined in Section~\ref{sec:Maximum likelihood estimation}. 
Further operational aspects are discussed in Section~\ref{sec:fas}.
In Section~\ref{subsec:Numerical evaluation of some properties of the model estimators}, properties of the estimators of the regression coefficients $\bbeta_j$ are evaluated through Monte Carlo experiments and compared to the estimators from the Gaussian CWM; a sensitivity study is also conducted in Section~\ref{subsec:Sensitivity study based on real data} based on a real data set.
The paper concludes with some discussion in Section~\ref{sec:Discussion and future work}.


\section{The model}
\label{sec:The model}

A contaminated Gaussian distribution, for a real-valued random vector $\bW$, is given by
\begin{equation}
f\left(\bw;\bmu_{\bW},\bSigma_{\bW},\alpha_{\bW},\eta_{\bW}\right)=\alpha_{\bW}\phi\left(\bw;\bmu_{\bW},\bSigma_{\bW}\right)+\left(1-\alpha_{\bW}\right)\phi\left(\bw;\bmu_{\bW},\eta_{\bW}\bSigma_{\bW}\right),
\label{eq:contaminated Gaussian distribution}
\end{equation}
where $\alpha_{\bW}\in\left(0,1\right)$ and $\eta_{\bW}>1$.
In \eqref{eq:contaminated Gaussian distribution}, $\eta_{\bW}$ denotes the degree of contamination, and because of the assumption $\eta_{\bW}>1$, it can be interpreted as the increase in variability due to the bad observations (i.e., it is an inflation parameter).
As a limiting case, when $\alpha_{\bW}\rightarrow 1^-$ and $\eta_{\bW}\rightarrow 1^+$, the Gaussian distribution is obtained.

A common and different way to burden the Gaussian tails (reference distribution), still maintaining ellipticity, is represented by the $t$ distribution (see \citealp{Kotz:Nada:Mult:2004} and \citealp{Lang:Litt:Tayl:Robu:1989} for details). 
An advantage of model~\eqref{eq:contaminated Gaussian distribution} with respect to the $t$ model is that, once the parameters are estimated, say $\hat{\bmu}_{\bW}$, $\hat{\bSigma}_{\bW}$, $\hat{\alpha}_{\bW}$, and $\hat{\eta}_{\bW}$, we can establish if a generic observation $\bw$ is either good or bad, with respect to the reference distribution, by means of the \textit{a~posteriori} probability
\begin{equation*}
P\left(\text{$\bw$ is good}\left|\hat{\bmu}_{\bW},\hat{\bSigma}_{\bW},\hat{\alpha}_{\bW},\hat{\eta}_{\bW}\right.\right)=\hat{\alpha}\phi\left(\bw;\hat{\bmu},\hat{\bSigma}\right)\Big/f\left(\bw;\hat{\bmu}_{\bW},\hat{\bSigma}_{\bW},\hat{\alpha}_{\bW},\hat{\eta}_{\bW}\right),
\end{equation*}
and $\bw$ will be considered good if $P\left(\text{$\bw$ is good}\left|\hat{\bmu}_{\bW},\hat{\bSigma}_{\bW},\hat{\alpha}_{\bW},\hat{\eta}_{\bW}\right.\right)\geq 1/2$, while it will be considered bad otherwise.

Based on model~\eqref{eq:contaminated Gaussian distribution}, \citet{Punz:McNi:Robu:2013} introduce, for robust model-based clustering, finite mixtures of contaminated Gaussian distributions with density
\begin{equation}
p\left(\bw;\bvartheta\right)=\sum_{j=1}^k\pi_jf\left(\bw;\bmu_{\bW|j},\bSigma_{\bW|j},\alpha_{\bW|j},\eta_{\bW|j}\right).
\label{eq:mixtures of contaminated Gaussian distributions}
\end{equation}
Unfortunately, with respect to the framework of this paper, model~\eqref{eq:mixtures of contaminated Gaussian distributions} does not account for local relations of the response $\bY$ on the covariate $\bX=\bx$ when $\bW=\left(\bX,\bY\right)$.
However, in a context of mixtures of regression models with fixed  covariates, the contaminated Gaussian distribution can be also considered to model $\bY|\bx$ in each mixture component; this leads to the mixture of contaminated Gaussian regression models
\begin{equation}
p\left(\by|\bx;\bvartheta\right)=\sum_{j=1}^k\pi_jf\left(\by;\bmu_{\bY}\left(\bx;\bbeta_j\right),\bSigma_{\bY|j},\alpha_{\bY|j},\eta_{\bY|j}\right).
\label{eq:mixtures of contaminated Gaussian regressions}
\end{equation}
However, because model~\eqref{eq:mixtures of contaminated Gaussian regressions} belongs to the class of mixtures of regression models with fixed covariates, it suffers from the assignment independence property.
Moreover, it can not be used to detect local leverage points (cf. Section~\ref{sec:Introduction}). 

To improve model~\eqref{eq:mixtures of contaminated Gaussian regressions}, we propose the contaminated Gaussian CWM; it is obtained by replacing the Gaussian distribution in model~\eqref{eq:Gaussian CWM} with the contaminated Gaussian distribution.
This yields
\begin{equation}
p\left(\bx,\by;\bvartheta\right)=\sum_{j=1}^k\pi_j
f\left(\by;\bmu_{\bY}\left(\bx;\bbeta_j\right),\bSigma_{\bY|j},\alpha_{\bY|j},\eta_{\bY|j}\right)
f\left(\bx;\bmu_{\bX|j},\bSigma_{\bX|j},\alpha_{\bX|j},\eta_{\bX|j}\right).
\label{eq:contaminated Gaussian CWM}
\end{equation}

\section{Relation with other contaminated models}
\label{sec:Relation with other contaminated models}

The contaminated Gaussian CWM defined in \eqref{eq:contaminated Gaussian CWM} can be related to the mixture-based contaminated models defined in Section~\ref{sec:The model}.

\subsection{Comparison with the mixture of contaminated Gaussian distributions}

To begin, we consider the comparison with mixtures of contaminated Gaussian distributions.
With this aim, it is convenient to write the parameters $\bmu_{\bW|j}$ and $\bSigma_{\bW|j}$, $j=1,\ldots,k$, of the mixture of contaminated Gaussian distributions defined in \eqref{eq:mixtures of contaminated Gaussian distributions} as
\begin{equation*}
\bmu_{\bW|j}=\begin{pmatrix}
\bmu_{\bX|j}\\
\bmu_{\bY|j}
\end{pmatrix}
\quad
\text{and}
\quad
\bSigma_{\bW|j}=\begin{pmatrix}
\bSigma_{\bX\bX|j} & \bSigma_{\bX\bY|j} \\
\bSigma_{\bY\bX|j} & \bSigma_{\bY\bY|j}
\end{pmatrix}.
\end{equation*}  
Based on well-known results about marginal and conditional distributions from a multivariate Gaussian random vector \citep[see, e.g.,][]{Mard:Kent:Bibb:Mult:1997}, model~\eqref{eq:mixtures of contaminated Gaussian distributions} can be rewritten as
\begin{align}
p\left(\bw;\bvartheta\right)=\sum_{j=1}^k\pi_j & f\left(\bw;\bmu_{\bW|j},\bSigma_{\bW|j},\alpha_{\bW|j},\eta_{\bW|j}\right)
=\sum_{j=1}^k\pi_j \left[\alpha_{\bW|j}\phi\left(\bw;\bmu_{\bW|j},\bSigma_{\bW|j}\right)+\left(1-\alpha_{\bW|j}\right)\phi\left(\bw;\bmu_{\bW|j},\eta_{\bW|j}\bSigma_{\bW|j}\right)\right]\nonumber\\
=\sum_{j=1}^k\pi_j & \Bigg[\frac{\alpha_{\bW|j}\phi\left(\bx;\bmu_{\bX|j},\bSigma_{\bX\bX|j}\right)}{\alpha_{\bW|j}\phi\left(\bx;\bmu_{\bX|j},\bSigma_{\bX\bX|j}\right)+\left(1-\alpha_{\bW|j}\right)\phi\left(\bx;\bmu_{\bX|j},\eta_{\bW|j}\bSigma_{\bX\bX|j}\right)}\phi\left(\by;\bmu_{\bY|\bX,j},\bSigma_{\bY\bY|j}\right)\nonumber\\
&+\frac{\left(1-\alpha_{\bW|j}\right)\phi\left(\bx;\bmu_{\bX|j},\eta_{\bW|j}\bSigma_{\bX\bX|j}\right)}{\alpha_{\bW|j}\phi\left(\bx;\bmu_{\bX|j},\bSigma_{\bX\bX|j}\right)+\left(1-\alpha_{\bW|j}\right)\phi\left(\bx;\bmu_{\bX|j},\eta_{\bW|j}\bSigma_{\bX\bX|j}\right)}\phi\left(\by;\bmu_{\bY|\bX,j},\eta_{\bW|j}\bSigma_{\bY\bY|j}\right)\Bigg]\nonumber\\
&\times\left[\alpha_{\bW|j}\phi\left(\bx;\bmu_{\bX|j},\bSigma_{\bX\bX|j}\right)+\left(1-\alpha_{\bW|j}\right)\phi\left(\bx;\bmu_{\bX|j},\eta_{\bW|j}\bSigma_{\bX\bX|j}\right)\right]\nonumber\\
=\sum_{j=1}^k\pi_j & \Big[\alpha_{\bW|j}\phi\left(\bx;\bmu_{\bX|j},\bSigma_{\bX\bX|j}\right)\phi\left(\by;\bmu_{\bY|\bX,j},\bSigma_{\bY\bY|j}\right)\nonumber\\
&+\left(1-\alpha_{\bW|j}\right)\phi\left(\bx;\bmu_{\bX|j},\eta_{\bW|j}\bSigma_{\bX\bX|j}\right)\phi\left(\by;\bmu_{\bY|\bX,j},\eta_{\bW|j}\bSigma_{\bY\bY|j}\right)\Big], 
\label{eq:rewrite mixtures of CNs}
\end{align}
where
\begin{equation*}
\bmu_{\bY|\bX,j}=\bmu_{\bY|j}+\bSigma_{\bY\bX|j}\bSigma_{\bX\bX|j}^{-1}\left(\bx-\bmu_{\bX|j}\right)
\end{equation*}
is a linear function of $\bx$.
For comparison's sake, it is also convenient to write model~\eqref{eq:contaminated Gaussian CWM} as
\begin{align}
p\left(\bx,\by;\bvartheta\right)=\sum_{j=1}^k\pi_j&
f\left(\by;\bmu_{\bY}\left(\bx;\bbeta_j\right),\bSigma_{\bY|j},\alpha_{\bY|j},\eta_{\bY|j}\right)
f\left(\bx;\bmu_{\bX|j},\bSigma_{\bX|j},\alpha_{\bX|j},\eta_{\bX|j}\right)\nonumber\\
=\sum_{j=1}^k\pi_j& \left[\alpha_{\bY|j}\phi\left(\by;\bmu_{\bY}\left(\bx;\bbeta_j\right),\bSigma_{\bY|j}\right)+\left(1-\alpha_{\bY|j}\right)\phi\left(\by;\bmu_{\bY}\left(\bx;\bbeta_j\right),\eta_{\bY|j}\bSigma_{\bY|j}\right)\right] \nonumber\\
&\times  \left[\alpha_{\bX|j}\phi\left(\bx;\bmu_{\bX|j},\bSigma_{\bX|j}\right)+\left(1-\alpha_{\bX|j}\right)\phi\left(\bx;\bmu_{\bX|j},\eta_{\bX|j}\bSigma_{\bX|j}\right)\right]\nonumber\\
=\sum_{j=1}^k\pi_j&\Big[\alpha_{\bX|j}\alpha_{\bY|j}\phi\left(\bx;\bmu_{\bX|j},\bSigma_{\bX|j}\right) \phi\left(\by;\bmu_{\bY}\left(\bx;\bbeta_j\right),\bSigma_{\bY|j}\right)\nonumber\\
&+\alpha_{\bX|j}\left(1-\alpha_{\bY|j}\right)\phi\left(\bx;\bmu_{\bX|j},\bSigma_{\bX|j}\right) \phi\left(\by;\bmu_{\bY}\left(\bx;\bbeta_j\right),\eta_{\bY|j}\bSigma_{\bY|j}\right)\nonumber\\
&+\left(1-\alpha_{\bX|j}\right)\alpha_{\bY|j}\phi\left(\bx;\bmu_{\bX|j},\eta_{\bX|j}\bSigma_{\bX|j}\right) \phi\left(\by;\bmu_{\bY}\left(\bx;\bbeta_j\right),\bSigma_{\bY|j}\right)\nonumber\\
&+\left(1-\alpha_{\bX|j}\right)\left(1-\alpha_{\bY|j}\right)\phi\left(\bx;\bmu_{\bX|j},\eta_{\bX|j}\bSigma_{\bX|j}\right) \phi\left(\by;\bmu_{\bY}\left(\bx;\bbeta_j\right),\eta_{\bY|j}\bSigma_{\bY|j}\right)\Big].
\label{eq:rewrite contaminated CWM}
\end{align}
Comparing the expressions enclosed within square brackets at the end of \eqref{eq:rewrite mixtures of CNs} with the equivalent term in \eqref{eq:rewrite contaminated CWM},
it is straightforward to realize the difference between the models.

\subsection{Comparison with the mixture of contaminated Gaussian regressions}

The second comparison concerns mixtures of contaminated Gaussian regression models as defined in \eqref{eq:mixtures of contaminated Gaussian regressions}.
The comparison is not direct because, while model~\eqref{eq:mixtures of contaminated Gaussian regressions} is defined on the conditional distribution $p\left(\by|\bx\right)$, model~\eqref{eq:contaminated Gaussian CWM} is defined on the joint distribution $p\left(\bx,\by\right)$.
Although we can not compute $p\left(\bx,\by\right)$ from a mixture of regression models with fixed covariates because this class of models does not consider modeling for the marginal distribution $p\left(\bx\right)$, we can still compute the conditional distribution $p\left(\by|\bx\right)$ from the contaminated Gaussian CWM.
In particular, by integrating out $\by$ from model~\eqref{eq:contaminated Gaussian CWM} we obtain
\begin{equation}
p\left(\bx;\bvartheta\right)=\sum_{j=1}^k\pi_j
f\left(\bx;\bmu_{\bX|j},\bSigma_{\bX|j},\alpha_{\bX|j},\eta_{\bX|j}\right);
\label{eq:marginal from a CN CWM}
\end{equation}
this is a mixture of contaminated Gaussian distributions for the $\bX$ only.
The ratio of \eqref{eq:contaminated Gaussian CWM} over \eqref{eq:marginal from a CN CWM} yields 
\begin{equation}
p\left(\by|\bx;\bvartheta\right)=
\sum_{j=1}^k
\frac{\pi_jf\left(\bx;\bmu_{\bX|j},\bSigma_{\bX|j},\alpha_{\bX|j},\eta_{\bX|j}\right)}{\displaystyle\sum_{h=1}^k\pi_hf\left(\bx;\bmu_{\bX|h},\bSigma_{\bX|h},\alpha_{\bX|h},\eta_{\bX|h}\right)}f\left(\by;\bmu_{\bY}\left(\bx;\bbeta_j\right),\bSigma_{\bY|j},\alpha_{\bY|j},\eta_{\bY|j}\right).
\label{eq:conditional from a CN CWM}
\end{equation} 
Model~\eqref{eq:conditional from a CN CWM} is the conditional distribution of $\bY|\bx$ from a contaminated Gaussian CWM; it can be seen as a mixture of regression models with (dynamic) weights depending on $\bx$.

The following proposition shows as the family of mixtures of contaminated Gaussian regression models can be seen as nested in the family of contaminated Gaussian CWMs, as defined by \eqref{eq:conditional from a CN CWM}.
\begin{pro}
\label{pro:contaminated Gaussian CWM versus mixtures of contaminated Gaussian regressions}
If, in \eqref{eq:conditional from a CN CWM}, $\bmu_{\bX|1}=\cdots=\bmu_{\bX|k}=\bmu_{\bX}$, $\bSigma_{\bX|1}=\cdots=\bSigma_{\bX|k}=\bSigma_{\bX}$, $\alpha_{\bX|1}=\cdots=\alpha_{\bX|k}=\alpha_{\bX}$, and $\eta_{\bX|1}=\cdots=\eta_{\bX|k}=\eta_{\bX}$, then mixtures of contaminated Gaussian regression models, as defined by \eqref{eq:mixtures of contaminated Gaussian regressions}, can be seen as a particular case of the contaminated Gaussian CWM, as defined by \eqref{eq:conditional from a CN CWM}.
\end{pro}
\begin{proof}
A proof of this proposition is provided in \ref{app:Proof of Proposition 1}.
\end{proof}

\section[Identifiability]{Identifiability}
\label{sec:Identifiability}

Before outlining maximum likelihood (ML) parameter estimation for model~\eqref{eq:contaminated Gaussian CWM}, it is important to establish its identifiability, that is, two sets of parameters in the model, say $\bvartheta$ and $\widetilde{\bvartheta}$, which do not agree after permutation cannot yield the same mixture distribution.
Identifiability is a necessary requirement, \textit{inter alia}, for the usual asymptotic theory to hold for ML estimation of the model parameters (cf.\ Section~\ref{sec:Maximum likelihood estimation}).

General conditions for the identifiability of mixtures of (linear Gaussian) regression models with fixed and random covariates are provided in \citet{Henn:Iden:2000}.
A sufficient condition for the identifiability of the mixture of contaminated Gaussian distributions is given in \citet{Punz:McNi:Robu:2013}.
These results will be used in Proposition~\ref{pro:1} to show that model~\eqref{eq:contaminated Gaussian CWM} is identifiable provided that all pairs $\left(\bbeta_j,\bSigma_{\bY|j}\right)$, $j=1,\ldots,k$, are pairwise distinct.
Note that, the positivity of all the weights $\pi_j$ avoids nonidentifiability due to empty components \citep[see][Section~1.3.3 for details]{Fruh:Fine:2006}.
\begin{pro}
\label{pro:1}
Let
\begin{displaymath}
p\left(\bx,\by;\bvartheta\right)=\sum_{j=1}^k\pi_j
f\left(\by;\bmu_{\bY}\left(\bx;\bbeta_j\right),\bSigma_{\bY|j},\alpha_{\bY|j},\eta_{\bY|j}\right)
f\left(\bx;\bmu_{\bX|j},\bSigma_{\bX|j},\alpha_{\bX|j},\eta_{\bX|j}\right).
\end{displaymath}
and 
\begin{displaymath}
p\left(\bx,\by;\widetilde{\bvartheta}\right)=\sum_{s=1}^{\widetilde{k}}\widetilde{\pi}_s
f\left(\by;\bmu_{\bY}\left(\bx;\widetilde{\bbeta}_s\right),\widetilde{\bSigma}_{\bY|s},\widetilde{\alpha}_{\bY|s},\widetilde{\eta}_{\bY|s}\right)
f\left(\bx;\widetilde{\bmu}_{\bX|s},\widetilde{\bSigma}_{\bX|s},\widetilde{\alpha}_{\bX|s},\widetilde{\eta}_{\bX|s}\right).
\end{displaymath}
be two different parameterizations of the contaminated Gaussian CWM given in \eqref{eq:contaminated Gaussian CWM}. 
If $j\neq l$, with $j,l\in\left\{1,\ldots,k\right\}$, implies
\begin{equation}
\left\|\bbeta_j-\bbeta_l\right\|_2^2+\left\|\bSigma_{\bY|j}-a\bSigma_{\bY|l}\right\|_2^2\neq 0
\label{eq:sufficient condition}
\end{equation}
for all $a>0$, where $\left\| \cdot \right\|_2$ is the Froebenius norm, then the equality $p\left(\bx,\by;\bvartheta\right)=p\left(\bx,\by;\widetilde{\bvartheta}\right)$, for almost all $\bx\in\real^{d_{\bX}}$, implies that $k=\widetilde{k}$ and also implies that for each $j\in\left\{1,\ldots,k\right\}$ there exists an $s\in\left\{1,\ldots,k\right\}$ such that $\pi_j=\widetilde{\pi}_s$, $\alpha_{\bX|j}=\widetilde{\alpha}_{\bX|s}$, $\bmu_{\bX|j}=\widetilde{\bmu}_{\bX|s}$, $\bSigma_{\bX|j}=\widetilde{\bSigma}_{\bX|s}$, $\eta_{\bX|j}=\widetilde{\eta}_{\bX|s}$, $\alpha_{\bY|j}=\widetilde{\alpha}_{\bY|s}$, $\bbeta_j=\widetilde{\bbeta}_s$, $\bSigma_{\bY|j}=\widetilde{\bSigma}_{\bY|s}$, and $\eta_{\bY|j}=\widetilde{\eta}_{\bY|s}$.
\end{pro}
\begin{proof}
A proof of this proposition is provided in \ref{app:Proof of Proposition 2}.
\end{proof}

\section{Maximum likelihood estimation}
\label{sec:Maximum likelihood estimation}

\subsection{An ECM algorithm}

Let 
$\left(\bx_1,\by_1\right),\ldots,\left(\bx_n,\by_n\right)$
be a sample from model~\eqref{eq:contaminated Gaussian CWM}.
To find ML estimates for the parameters of this model, we adopt the expectation-conditional maximization (ECM) algorithm of \citet{Meng:Rubin:Maxi:1993}. 
The ECM algorithm is a variant of the classical expectation-maximization (EM) algorithm \citep{Demp:Lair:Rubi:Maxi:1977}, which is a natural approach for ML estimation when data are incomplete. 
In our case, there are three sources of incompleteness.
The first source, the classical one in the use of mixture models, arise from the fact that for each observation we do not know its component membership; this source is governed by an indicator vector $\bz_i=\left(z_{i1},\ldots,z_{ik}\right)$, where $z_{i1}=1$ if $\left(\bx_i,\by_i\right)$ comes from component $j$ and $z_{ij}=0$ otherwise.
The other two sources, which are specific for this model, arise from the fact that for each observation we do not know if it is an outlier and/or a leverage point with reference to component $j$ (cf. \tablename~\ref{tab:Atypical observation labelling}).
To denote these sources of incompleteness, we use $\bu_i=\left(u_{i1},\ldots,u_{ik}\right)$, where $u_{ij}=1$ if $\left(\bx_i,\by_i\right)$ is not an outlier in component $j$ and $u_{ij}=0$ otherwise, and $\bv_i=\left(v_{i1},\ldots,v_{ik}\right)$, where $v_{ij}=1$ if $\left(\bx_i,\by_i\right)$ is not a leverage point in component $j$ and $v_{ij}=0$ otherwise. 
Therefore, complete-data likelihood can be written
\begin{align*}
L_c\left(\bvartheta\right)=\prod_{i=1}^n\prod_{j=1}^k\Bigg\{ & \pi_j    
\left[\alpha_{\bX|j}\phi\left(\bx;\bmu_{\bX|j},\bSigma_{\bX|j}\right)\right]^{v_{ij}}\left[\left(1-\alpha_{\bX|j}\right)\phi\left(\by;\bmu_{\bX|j},\eta_{\bX|j}\bSigma_{\bX|j}\right)\right]^{1-v_{ij}}\nonumber\\
& 
\times
\left[\alpha_{\bY|j}\phi\left(\by;\bmu_{\bY}\left(\bx_i;\bbeta_j\right),\bSigma_{\bY|j}\right)\right]^{u_{ij}}\left[\left(1-\alpha_{\bY|j}\right)\phi\left(\by;\bmu_{\bY}\left(\bx_i;\bbeta_j\right),\eta_{\bY|j}\bSigma_{\bY|j}\right)\right]^{1-u_{ij}}
\Bigg\}^{z_{ij}}.
\end{align*}
Therefore, the complete-data log-likelihood, which is the core of the algorithm, becomes
\begin{equation}
l_c\left(\bvartheta\right)
=l_{1c}\left(\bpi\right)
+l_{2c}\left(\balpha_{\bX}\right)
+l_{3c}\left(\bmu_{\bX},\bSigma_{\bX},\boldsymbol{\eta}_{\bX}\right)
+l_{4c}\left(\balpha_{\bY}\right)
+l_{5c}\left(\bbeta,\bSigma_{\bY},\boldsymbol{\eta}_{\bY}\right),
\label{eq:complete-data log-likelihood}
\end{equation}
where $\bpi=\left(\pi_1,\ldots,\pi_k\right)$, $\bmu_{\bX}=\left(\bmu_{\bX|1},\ldots,\bmu_{\bX|k}\right)$, $\bSigma_{\bX}=\left(\bSigma_{\bX|1},\ldots,\bSigma_{\bX|k}\right)$, $\balpha_{\bX}=\left(\alpha_{\bX|1},\ldots,\alpha_{\bX|k}\right)$, $\boldsymbol{\eta}_{\bX}=\left(\eta_{\bX|1},\ldots,\eta_{\bX|k}\right)$, $\bbeta=\left(\bbeta_1,\ldots,\bbeta_k\right)$, $\bSigma_{\bY}=\left(\bSigma_{\bY|1},\ldots,\bSigma_{\bY|k}\right)$, $\balpha_{\bY}=\left(\alpha_{\bY|1},\ldots,\alpha_{\bY|k}\right)$, $\boldsymbol{\eta}_{\bY}=\left(\eta_{\bY|1},\ldots,\eta_{\bY|k}\right)$, 
\begin{align*}
&l_{1c}\left(\bpi\right)=\sum_{i=1}^{n}\sum_{j=1}^{k}{z}_{ij}\ln \pi_j,\qquad\qquad
l_{2c}\left(\balpha_{\bY}\right)=\sum_{i=1}^{n}\sum_{j=1}^{k}z_{ij}\left[v_{ij}\ln \alpha_{\bX|j}+\left(1-v_{ij}\right)\ln \left(1-\alpha_{\bX|j}\right)\right],\\
&l_{3c}\left(\bmu_{\bX},\bSigma_{\bX},\boldsymbol{\eta}_{\bX}\right)=-\frac{1}{2}\sum_{i=1}^n\sum_{j=1}^k\Biggl\{z_{ij}\ln\left|\bSigma_{\bX|j}\right|+d_{\bX}z_{ij}\left(1-v_{ij}\right)\ln\eta_{\bX|j}+z_{ij}\left(v_{ij}+\frac{1-v_{ij}}{\eta_{\bX|j}}\right)\delta\left(\bx_i,\bmu_{\bX|j};\bSigma_{\bX|j}\right)\Biggr\},\\
&l_{4c}\left(\balpha_{\bY}\right)=\sum_{i=1}^{n}\sum_{j=1}^{k}z_{ij}\left[u_{ij}\ln \alpha_{\bY|j}+\left(1-u_{ij}\right)\ln \left(1-\alpha_{\bY|j}\right)\right],\\
&l_{5c}\left(\bbeta,\bSigma_{\bY},\boldsymbol{\eta}_{\bY}\right)=-\frac{1}{2}\sum_{i=1}^n\sum_{j=1}^k\Biggl\{z_{ij}\ln\left|\bSigma_{\bY|j}\right|+d_{\bY}z_{ij}\left(1-u_{ij}\right)\ln\eta_{\bY|j}+z_{ij}\left(u_{ij}+\frac{1-u_{ij}}{\eta_{\bY|j}}\right)\delta\left(\bx_i,\bmu_{\bY}\left(\bx_i;\bbeta_j\right);\bSigma_{\bY|j}\right)\Biggr\},
\end{align*}
and where $\delta\left(\bw,\bmu;\bSigma\right)=\left(\bw-\bmu\right)'\bSigma^{-1}\left(\bw-\bmu\right)$ denotes the squared Mahalanobis distance between $\bw$ and $\boldsymbol{\mu}$, with covariance matrix $\boldsymbol{\Sigma}$.
The ECM algorithm iterates between three steps, an E-step and two CM-steps, until convergence. 
The only difference from the EM algorithm is that each M-step is replaced by two simpler CM-steps. 
They arise from the partition $\bvartheta=\left(\bvartheta_1,\bvartheta_2\right)$, where $\bvartheta_1=\left(\bpi,\bmu_{\bX},\bSigma_{\bX},\balpha_{\bX},\bbeta,\bSigma_{\bY},\balpha_{\bY}\right)$ and $\bvartheta_2=\left(\boldsymbol{\eta}_{\bX},\boldsymbol{\eta}_{\bY}\right)$.

\subsubsection{E-step.}
\label{subsubsec:E-step}

The E-step, on the $\left(r+1\right)$th iteration of the ECM algorithm, requires the calculation of $Q(\bvartheta|\bvartheta^{\left(r\right)})$, the current conditional expectation of $l_c\left(\bvartheta\right)$.
To do this, we need to calculate $E_{\bvartheta^{\left(r\right)}}\left(Z_{ij}|\bx_i,\by_i\right)$, $E_{\bvartheta^{\left(r\right)}}\left(V_{ij}|\bx_i,\bz_i\right)$, and $E_{\bvartheta^{\left(r\right)}}\left(U_{ij}|\by_i,\bz_i\right)$, $i=1,\ldots,n$ and $j=1,\ldots,k$.
They are respectively given by
\begin{displaymath}
E_{\bvartheta^{\left(r\right)}}\left(Z_{ij}|\bx_i,\by_i\right)=\frac{\pi_j^{\left(r\right)}f\left(\by_i;\bmu_{\bY}\left(\bx_i;\bbeta_j^{\left(r\right)}\right),\bSigma_{\bY|j}^{\left(r\right)},\alpha_{\bY|j}^{\left(r\right)},\eta_{\bY|j}^{\left(r\right)}\right)
f\left(\bx_i;\bmu_{\bX|j}^{\left(r\right)},\bSigma_{\bX|j}^{\left(r\right)},\alpha_{\bX|j}^{\left(r\right)},\eta_{\bX|j}^{\left(r\right)}\right)}{p\left(\bx_i,\by_i;\bvartheta^{\left(r\right)}\right)}\eqqcolon z_{ij}^{\left(r\right)},
\end{displaymath}
\begin{equation}
E_{\bvartheta^{\left(r\right)}}\left(V_{ij}|\bx_i,\bz_i\right)=\frac{\alpha_{\bX|j}^{\left(r\right)}\phi\left(\bx_i;\bmu_{\bX|j}^{\left(r\right)},\bSigma_{\bX|j}^{\left(r\right)}\right)}{f\left(\bx_i;\bmu_{\bX|j}^{\left(r\right)},\bSigma_{\bX|j}^{\left(r\right)},\alpha_{\bX|j}^{\left(r\right)},\eta_{\bX|j}^{\left(r\right)}\right)}\eqqcolon v_{ij}^{\left(r\right)},
\label{eq:v update}	
\end{equation}
and
\begin{equation}
E_{\bvartheta^{\left(r\right)}}\left(U_{ij}|\by_i,\bz_i\right)=\frac{\alpha_{\bY|j}^{\left(r\right)}\phi\left(\by_i;\bmu_{\bY}\left(\bx_i;\bbeta_j^{\left(r\right)}\right),\bSigma_{\bY|j}^{\left(r\right)}\right)}{f\left(\by_i;\bmu_{\bY}\left(\bx_i;\bbeta_j^{\left(r\right)}\right),\bSigma_{\bY|j}^{\left(r\right)},\alpha_{\bY|j}^{\left(r\right)},\eta_{\bY|j}^{\left(r\right)}\right)}\eqqcolon u_{ij}^{\left(r\right)}.	
\label{eq:u update}	
\end{equation}
Then, by substituting $z_{ij}$ with $z_{ij}^{\left(r\right)}$, $v_{ij}$ with $v_{ij}^{\left(r\right)}$, and $u_{ij}$ with $u_{ij}^{\left(r\right)}$ in \eqref{eq:complete-data log-likelihood}, we obtain $Q\left(\bvartheta|\bvartheta^{\left(r\right)}\right)$; see \ref{app:Updates} for details.

\subsubsection{CM-step 1.}
\label{subsubsec:CM-step 1}

The first CM-step on the $\left(r+1\right)$th iteration of the ECM algorithm requires the calculation of $\bvartheta_1^{\left(r+1\right)}$ as the value of $\bvartheta_1$ that maximizes $Q\left(\bvartheta_1|\bvartheta^{\left(r\right)}\right)$ with $\bvartheta_2$ fixed at $\bvartheta_2^{\left(r\right)}$.
In particular, after some algebra, we obtain
\begin{align}
\pi_j^{\left(r+1\right)} & =\frac{n_j^{\left(r\right)}}{n},\qquad\qquad
\alpha_{\bX|j}^{\left(r+1\right)} =\frac{1}{n_j^{\left(r\right)}}\sum_{i=1}^n z_{ij}^{\left(r\right)}v_{ij}^{\left(r\right)},\nonumber\\
\bmu_{\bX|j}^{\left(r+1\right)} & =\frac{1}{s_j^{\left(r\right)}}\sum_{i=1}^n{z}_{ij}^{\left(r\right)}\left(v_{ij}^{\left(r\right)}+\frac{1-v_{ij}^{\left(r\right)}}{\eta_{\bX|j}^{\left(r\right)}}\right)\bx_i,\label{eq:muX update}\\
\bSigma_{\bX|j}^{\left(r+1\right)} &=\frac{1}{n_j^{\left(r\right)}}\sum_{i=1}^nz_{ij}^{\left(r\right)}\left(v_{ij}^{\left(r\right)}+\frac{1-v_{ij}^{\left(r\right)}}{\eta_{\bX|j}^{\left(r\right)}}\right)\left(\bx_i-\displaystyle\bmu_{\bX|j}^{\left(r+1\right)}\right)\left(\bx_i-\displaystyle\bmu_{\bX|j}^{\left(r+1\right)}\right)',\label{eq:SigmaX update}\\	
\alpha_{\bY|j}^{\left(r+1\right)} & =\frac{1}{n_j^{\left(r\right)}}\sum_{i=1}^n{z}_{ij}^{\left(r\right)}\left(v_{ij}^{\left(r\right)}+\frac{1-v_{ij}^{\left(r\right)}}{\eta_{\bY|j}^{\left(r\right)}}\right),\nonumber\\
\bbeta_j^{\left(r+1\right)} & = \left[\sum_{i=1}^n z_{ij}^{\left(r\right)}\left(u_{ij}^{\left(r\right)}+\frac{1-u_{ij}^{\left(r\right)}}{\eta_{\bY|j}^{\left(r\right)}}\right)\bx_i^*\bx_i^{*'}\right]^{-1}\left[\sum_{i=1}^n z_{ij}^{\left(r\right)}\left(u_{ij}^{\left(r\right)}+\frac{1-u_{ij}^{\left(r\right)}}{\eta_{\bY|j}^{\left(r\right)}}\right)\bx_i^*\by_i\right],\label{eq:beta update}\\
\bSigma_{\bY|j}^{\left(r+1\right)} &=\frac{1}{n_j^{\left(r\right)}}\sum_{i=1}^nz_{ij}^{\left(r\right)}\left(u_{ij}^{\left(r\right)}+\frac{1-u_{ij}^{\left(r\right)}}{\eta_{\bY|j}^{\left(r\right)}}\right)\left[\by_i-\displaystyle\bmu_{\bY}\left(\bx_i;\bbeta_j^{\left(r+1\right)}\right)\right]\left[\by_i-\displaystyle\bmu_{\bY}\left(\bx_i;\bbeta_j^{\left(r+1\right)}\right)\right]',\label{eq:SigmaY update}	
\end{align}
where 
\begin{displaymath}
	s_j^{\left(r\right)}=\sum_{i=1}^nz_{ij}^{\left(r\right)}\left(v_{ij}^{\left(r\right)}+\frac{1-v_{ij}^{\left(r\right)}}{\eta_{\bX|j}^{\left(r\right)}}\right)
\end{displaymath}
and $\displaystyle n_j^{\left(r\right)}=\sum_{i=1}^nz_{ij}^{\left(r\right)}$.
Details on the updates for \eqref{eq:muX update}, \eqref{eq:SigmaX update}, \eqref{eq:beta update}, and \eqref{eq:SigmaY update} are reported in \ref{app:Updates}.

\subsubsection{CM-step 2.}
\label{subsubsec:CM-step 2}

The second CM-step, on the $\left(r+1\right)$th iteration of the ECM algorithm, requires the calculation of $\bvartheta_2^{\left(r+1\right)}$ as the value of $\bvartheta_2$ that maximizes $Q\left(\bvartheta|\bvartheta^{\left(r\right)}\right)$ with $\bvartheta_1$ fixed at $\bvartheta_1^{\left(r+1\right)}$.
In particular, for each $j=1,\ldots,k$, we have to maximize
\begin{equation}
-\frac{d_{\bX}}{2}\sum_{i=1}^nz_{ij}^{\left(r\right)}\left(1-v_{ig}^{\left(r\right)}\right)\ln \eta_{\bX|j}
-\frac{1}{2}\sum_{i=1}^n{z}_{ij}^{\left(r\right)}\frac{1-v_{ij}^{\left(r\right)}}{\eta_{\bX|j}}\delta\left(\bx_i,\bmu_{\bX|j}^{\left(r+1\right)};\bSigma_{\bX|j}^{\left(r+1\right)}\right),
\label{eq:maximization function for etaXj}
\end{equation}
with respect to $\eta_{\bX|j}$, under the constraint $\eta_{\bX|j}>1$, and
\begin{equation}
-\frac{d_{\bY}}{2}\sum_{i=1}^nz_{ij}^{\left(r\right)}\left(1-u_{ig}^{\left(r\right)}\right)\ln \eta_{\bY|j}
-\frac{1}{2}\sum_{i=1}^n{z}_{ij}^{\left(r\right)}\frac{1-u_{ij}^{\left(r\right)}}{\eta_{\bY|j}}\delta\left(\by_i,\bmu_{\bY}\left(\bx_i;\bbeta_j^{\left(r+1\right)}\right);\bSigma_{\bY|j}^{\left(r+1\right)}\right),
\label{eq:maximization function for etaYj}
\end{equation}
with respect to $\eta_{\bY|j}$, under the constraint $\eta_{\bY|j}>1$.
Operationally, the \texttt{optimize()} function in the \texttt{stats} package for \textsf{R} \citep{R} is used to perform a numerical search of the maximum of \eqref{eq:maximization function for etaXj} and \eqref{eq:maximization function for etaYj} over the interval $\left(1,\eta^*\right)$, with $\eta^*>1$.
In the analyses of Section~\ref{sec:Numerical studies}, we fix $\eta^*=500$ to facilitate faster convergence.

\subsection{Computational aspects}
\label{subsec:Computational aspects}

Code for the ECM algorithm was written in \textsf{R} and it is available from the authors upon request.
Further aspects related to the implementation of the algorithm are described in the following.

\subsubsection{Initialization}
\label{subsubsec:Initialization}

The choice of the starting values for EM-based algorithms constitutes an important issue (see, e.g., \citealp{Bier:Cele:Gova:Choo:2003}, \citealp{Karl:Xeka:Choo:2003}, and \citealp{Bagn:Punz:Fine:2013}).
For the ECM algorithm described before, two natural strategies are:  
\begin{enumerate}
	\item choosing initial values $z_{ij}^{\left(0\right)}$, $v_{ij}^{\left(0\right)}$, and $u_{ij}^{\left(0\right)}$, respectively for $z_{ij}$, $v_{ij}$, and $u_{ij}$, $i=1,\ldots,n$ and $j=1,\ldots,k$, in the E-step of the first iteration;
	\item selecting an initial value $\bvartheta^{\left(0\right)}$ for $\bvartheta$ in the two CM-steps of the first iteration.
\end{enumerate}
By considering the first strategy, we suggest the following technique.
The $k$-component Gaussian CWM in \eqref{eq:Gaussian CWM} can be seen as nested in the $k$-component contaminated Gaussian CWM in \eqref{eq:contaminated Gaussian CWM} when $\alpha_{\bX|j},\alpha_{\bY|j}\rightarrow 1^-$ and $\eta_{\bX|j},\eta_{\bY|j}\rightarrow 1^+$, $j=1,\ldots,k$.
Under these conditions, $u_{ij},v_{ij}\rightarrow 1^-$, $i=1,\ldots,n$ and $j=1,\ldots,k$, and model~\eqref{eq:contaminated Gaussian CWM} tends to model~\eqref{eq:Gaussian CWM}.
Then, the posterior probabilities from the EM algorithm for the Gaussian CWM \citep[as described by][]{Dang:Punz:McNi:Ingr:Brow:Mult:2014}, along with the constraints $u_{ij}^{\left(0\right)}=v_{ij}^{\left(0\right)}=w^{\left(0\right)}$, $w^{\left(0\right)}\rightarrow 1^-$, $i=1,\ldots,n$, and $j=1,\ldots,k$, can be used to initialize the first E-step of our ECM algorithm.
From an operational point of view, thanks to the monotonicity property of the ECM algorithm \citep[see, e.g.,][p.~28]{McLa:Kris:TheE:2007}, this also guarantees that the observed-data log-likelihood of the contaminated Gaussian CWM will be always greater than, or equal to, the observed-data log-likelihood of the ``starting'' Gaussian CWM.
This is a fundamental consideration for the use of likelihood-based model selection criteria for choosing between a Gaussian CWM and a contaminated Gaussian CWM.   

In the analyses of Section~\ref{sec:Numerical studies}, $w^{\left(0\right)}=0.999$ and the E-step of the EM algorithm for the Gaussian CWM is initialized based on the posterior probabilities arising from the fitting of an unconstrained $k$-component mixture of Gaussian distributions for $\bW=\left(\bX,\bY\right)$, as implemented by the \texttt{Mclust()} function of the \textbf{mclust} package for {\sf R} \citep{Fral:Raft:Murp:Scru:mclu:2012}.

\subsubsection{Convergence criterion}
\label{subsubsec:Convergence criterion}

The Aitken acceleration \citep{Aitk:OnBe:1926} is used to estimate the asymptotic maximum of the log-likelihood at each iteration of the ECM algorithm. 
Based on this estimate, we can decide whether or not the algorithm has reached convergence;
i.e., whether or not the log-likelihood is sufficiently close to its estimated asymptotic value. 
The Aitken acceleration at iteration $r+1$ is given by
\begin{equation*}
a^{\left(r+1\right)}=\frac{l^{\left(r+2\right)}-l^{\left(r+1\right)}}{l^{\left(r+1\right)}-l^{\left(r\right)}},
\end{equation*}
where $l^{\left(r\right)}$ is the observed-data log-likelihood value from iteration $r$. 
Then, the asymptotic estimate of the log-likelihood at iteration $r + 2$ is given by
\begin{equation*}
l_{\infty}^{\left(r+2\right)}=l^{\left(r+1\right)}+\frac{1}{1-a^{\left(r+1\right)}}\left(l^{\left(r+2\right)}-l^{\left(r+1\right)}\right);
\end{equation*}
cf.\ \citet{Bohn:Diet:Scha:Schl:Lind:TheD:1994}.
The ECM algorithm can be considered to have converged when $l_{\infty}^{\left(r+2\right)}-l^{\left(r+1\right)}<\epsilon$. 
In the analyses of Section~\ref{sec:Numerical studies}, $\epsilon=0.0001$.

\section{Operational aspects}
\label{sec:fas}

\subsection{Some notes on robustness}
\label{subsec:Some notes on robustness}

Based on \eqref{eq:muX update}, $\bmu_{\bX|j}^{\left(r+1\right)}$ is a weighted mean of the $\bx_i$ values, with weights depending on 
\begin{equation}
v_{ij}^{\left(r\right)}+\frac{1-v_{ij}^{\left(r\right)}}{\eta_{\bX|j}^{\left(r\right)}}.
\label{eq:downweight for X}
\end{equation}
Analogously, based on \eqref{eq:beta update}, the regression coefficients $\bbeta_j^{\left(r+1\right)}$ can be considered a weighted least squares estimate with weights depending on 
\begin{equation}
u_{ij}^{\left(r\right)}+\frac{1-u_{ij}^{\left(r\right)}}{\eta_{\bY|j}^{\left(r\right)}}.
\label{eq:downweight for Y}
\end{equation}
It is easy to note that \eqref{eq:downweight for X} and \eqref{eq:downweight for Y} have the same structure.
Based on \eqref{eq:v update} and \eqref{eq:u update}, also the structure of the updates for $v_{ij}^{\left(r\right)}$ and $u_{ij}^{\left(r\right)}$ is the same.
Now, consider these updates as a function of the squared Mahalanobis distance (i.e., the squared standardized residuals) $\delta$; the common updating function in \eqref{eq:v update} and \eqref{eq:u update} can be so written as
\begin{equation}
g\left(\delta;\alpha,\eta\right)=\frac{\alpha\exp\left\{-\frac{\delta}{2}\right\}}{\alpha\exp\left\{-\frac{\delta}{2}\right\}+\frac{\left(1-\alpha\right)}{\sqrt{\eta}}\exp\left\{-\frac{\delta}{2\eta}\right\}}=\frac{1}{1+\frac{\left(1-\alpha\right)}{\alpha}\frac{1}{\sqrt{\eta}}\exp\left\{\frac{\delta}{2}\left(1-\frac{1}{\eta}\right)\right\}},
\label{eq:updating function for u and v}
\end{equation}
with $\delta\geq 0$.
Due to the constraint $\eta>1$, from the last expression of \eqref{eq:updating function for u and v} it is straightforward to realize that $g\left(\delta;\alpha,\eta\right)$ is a decreasing function of $\delta$. 
Based on \eqref{eq:updating function for u and v}, formulas \eqref{eq:downweight for X} and \eqref{eq:downweight for Y} can be written as  
\begin{equation}
w\left(\delta;\alpha,\eta\right)=g\left(\delta;\alpha,\eta\right)+\frac{1-g\left(\delta;\alpha,\eta\right)}{\eta}=\frac{1}{\eta}\left[1+\left(\eta-1\right)g\left(\delta;\alpha,\eta\right)\right].
\label{eq:component of the down-weighting}
\end{equation}
From the last expression of \eqref{eq:component of the down-weighting}, it easy to realize that $w\left(\delta;\alpha,\eta\right)$ is an increasing function of $g\left(\delta;\alpha,\eta\right)$; this also means that $w\left(\delta;\alpha,\eta\right)$ is a decreasing function of $\delta$.
Therefore, the weights in \eqref{eq:downweight for X} and \eqref{eq:downweight for Y} reduce, respectively, the effect of leverage points in the estimation of $\bmu_{\bX|j}$ and the effect of outliers in the estimation of $\bbeta_j$, so providing a robust way to estimate $\bmu_{\bX|j}$ and $\bbeta_j$, $j=1,\ldots,k$.
In addition, from \eqref{eq:SigmaX update} and \eqref{eq:SigmaY update}, the larger squared residuals $\delta$ also have smaller effects on $\bSigma_{\bX|j}$ and $\bSigma_{\bY|j}$, $j=1,\ldots,k$, due to the weights in \eqref{eq:downweight for X} and \eqref{eq:downweight for Y}, respectively. 
See \citet{Litt:Robu:1988} for a discussion on down-weighting of the atypical observations for the contaminated Gaussian distribution. 


\subsection{Automatic detection of atypical points}
\label{subsec:Automatic detection of noise}

For a contaminated Gaussian CWM, the classification of an observation $\left(\bx_i,\by_i\right)$ means: 
\begin{description}
	\item[Step 1.] determine its component of membership;
	\item[Step 2.] establish if it is typical, outlier, good leverage, or bad leverage in that component (cf.~\tablename~\ref{tab:Atypical observation labelling}).
\end{description}
Let $\hat{\bu}_i$, $\hat{\bv}_i$, and $\hat{\bz}_i$ denote, respectively, the expected values of $\bu_i$, $\bv_i$, and $\bz_i$ arising from the ECM algorithm, i.e., $\hat{u}_{ij}$, $\hat{v}_{ij}$, and $\hat{z}_{ij}$ are the values of $u_{ij}$, $v_{ij}$, and $z_{ij}$, respectively, at convergence.
To evaluate the component membership of $\left(\bx_i,\by_i\right)$, we use the maximum \textit{a~posteriori} probabilities (MAP) operator
$$
\text{MAP}\left(\hat{z}_{ij}\right)=
\begin{cases}
1 & \text{if } \max_h\{\hat{z}_{ih}\} \text{ occurs in component $h=j$,}\\
0 & \text{otherwise}.
\end{cases}
$$
We then consider $\hat{u}_{ih}$ and $\hat{v}_{ih}$, where $h$ is selected such that $\text{MAP}\left(\hat{z}_{ih}\right)=1$.
Although $\left(1-\hat{u}_{ih}\right)$ and $\left(1-\hat{v}_{ih}\right)$ provide the richest information about the probability that $\left(\bx_i,\by_i\right)$ is an outlier or a leverage point, respectively, in group $h$, the user could be interested in obtaining a classification of this observation according to \tablename~\ref{tab:Atypical observation labelling}.
In such a case, the rule given in \tablename~\ref{tab:practical detection} could be applied.
\begin{table}[!ht]
\caption{
\label{tab:practical detection}
Rule for classifying a generic observation $\left(\bx_i,\by_i\right)$ in one of the four categories of \tablename~\ref{tab:Atypical observation labelling}.
}
\centering
\begin{tabular}{ccc}
\toprule
\backslashbox{$\hat{u}_{ih}$}{$\hat{v}_{ih}$} 	&	$\left[0,0.5\right)$	&	$\left[0.5,1\right]$ \\
\midrule															
$\left[0,0.5\right)$    &  bad leverage   & outlier \\
$\left[0.5,1\right]$    &  good leverage  & typical (bulk of the data)  \\
\bottomrule	
\end{tabular}
\end{table} 

Thus, once the observation has been classified in one of the $k$ groups, the approach reveals richer information about the role of that observation in that group. 
Note also that, the resulting information from \tablename~\ref{tab:practical detection} can be used to eventually eliminate some of the atypical observations (such as outliers and bad leverage points) if such an outcome is desired \citep{Berk:Bent:Esti:1988}. 


\subsection{Constraints for detection of atypical points}
\label{subsec:Constraints for detection of outliers}

When the contaminated Gaussian CWM is used for detection of atypical points in each group, $\left(1-\alpha_{\bX|j}\right)$ and $\left(1-\alpha_{\bY|j}\right)$ represent the proportion of leverage points and outliers, respectively.
As suggested by \citet{Punz:McNi:Robu:2013}, for these parameters one could require that in the $j$th group, $j=1,\ldots,k$, the proportion of typical observations, with respect to $\bX$ and $\bY$, separately, is at least equal to a pre-determined value $\alpha^*$. 
In this case, the \texttt{optimize()} function is also used for a numerical search of the maximum $\alpha_{\bX|j}^{\left(r+1\right)}$, over the interval $\left(\alpha^*,1\right)$, of the function
\begin{displaymath}
\sum_{i=1}^nz_{ij}^{\left(r\right)}\left[v_{ij}^{\left(r\right)}\ln \alpha_{\bX|j}+\left(1-v_{ij}^{\left(r\right)}\right)\ln \left(1-\alpha_{\bX|j}\right)\right],	
\end{displaymath}
and of the maximum $\alpha_{\bY|j}^{\left(r+1\right)}$, over the interval $\left(\alpha^*,1\right)$, of the function
\begin{displaymath}
\sum_{i=1}^nz_{ij}^{\left(r\right)}\left[u_{ij}^{\left(r\right)}\ln \alpha_{\bY|j}+\left(1-u_{ij}^{\left(r\right)}\right)\ln \left(1-\alpha_{\bY|j}\right)\right].	
\end{displaymath}
In the analyses herein (cf. Section~\ref{sec:Numerical studies}), we use this approach to update $\alpha_{\bX|j}$ and $\alpha_{\bY|j}$ and we take $\alpha^*=0.5$.
Note that it is possible to fix $\alpha_{\bX|j}$ and $\alpha_{\bY|j}$ \textit{a~priori}.
This is somewhat analogous to the clusterwise linear regression through trimming approach, where one must to specify the proportion of outliers and leverage points in advance \citep[cf.][]{Garc:Gord:Mayo:SanM:Robu:2010}. 
However, pre-specifying points as outliers and/or leverage \textit{a~priori} may not be realistic in many practical scenarios.

%
%

\subsection{Choosing the number of mixture components}
\label{subsec:BIC}

The contaminated Gaussian CWM, in addition to $\bvartheta$, is also characterized by the number of components $k$. 
Thus far, this quantity has been treated as \textit{a~priori} fixed; nevertheless, for practical purposes, its selection is usually required.
One way (the usual way) to select $k$ is via computation of a convenient (likelihood-based) model selection criterion over a reasonable range of values for $k$, and then choosing the value of $k$ associated with the best value of the adopted criterion.
As in \citet{Punz:McNi:Robu:2013}, in the data analyses of Section~\ref{sec:Numerical studies} we will adopt the Bayesian information criterion \citep{Schw:Esti:1978}, i.e.,
\begin{displaymath}
\text{BIC}=2l\left(\hat{\bvartheta}\right)-m\ln n,
\end{displaymath}
where $m$ is the overall number of free parameters in the model. 

\section{Numerical studies}
\label{sec:Numerical studies}

In this section we evaluate the performance of the proposed model through Monte Carlo experiments performed using \textsf{R}.

\subsection{Evaluation of some properties of the estimators of the local regression coefficients}
\label{subsec:Numerical evaluation of some properties of the model estimators}

Properties of the estimators of the regression coefficients $\bbeta_j$, $j=1,\ldots,k$, are here evaluated through Monte Carlo experiments and compared to the estimators from the Gaussian CWM.
Our main interest is the effect of local atypical points, as conceived by the contaminated Gaussian CWM, on the bias and mean square error (MSE) of the estimators of $\bbeta_j$, $j=1,\ldots,k$, for the Gaussian CWM.

The following two scenarios of experiments are considered:
\begin{description}
	\item[Scenario A:] data generated from the Gaussian CWM;
	\item[Scenario B:] data generated from the contaminated Gaussian CWM.
\end{description}
Regardless from the considered scenario, the dimensions are $d_{\bX}=d_{\bY}=2$ and the number of mixture components is $k=2$.
The generating parameters of Scenario A are
\begin{equation}
	\pi_1=0.3,
	\quad \bmu_{\bX|1}=\begin{pmatrix*}[r]
-5	\\
-5	 
\end{pmatrix*},
	\quad \bSigma_{\bX|1}=\begin{pmatrix*}[r]
1 & 0	\\
0 & 1	 
\end{pmatrix*}
\quad \bbeta_1=\begin{pmatrix*}[r]
-2 & -2	\\
-1 &  1 \\
 1 & -1 
\end{pmatrix*},
	\quad \text{and} \quad \bSigma_{\bY|1}=\begin{pmatrix*}[l]
0.4 & 0	\\
0 & 0.4	 
\end{pmatrix*},
\label{eq:generating parameters in group 1}
\end{equation}
for the first mixture component, and 
\begin{equation}
	\pi_2=0.7,
	\quad \bmu_{\bX|2}=\begin{pmatrix*}[r]
5	\\
5	 
\end{pmatrix*},
	\quad \bSigma_{\bX|2}=\begin{pmatrix*}[r]
1 & 0	\\
0 & 1	 
\end{pmatrix*}
\quad \bbeta_2=\begin{pmatrix*}[r]
2 &  2	\\
1 & -1 \\
-1 & 1 
\end{pmatrix*},
	\quad \text{and} \quad \bSigma_{\bY|2}=\begin{pmatrix*}[l]
0.4 & 0	\\
0 & 0.4	 
\end{pmatrix*},
\label{eq:generating parameters in group 2}
\end{equation}
for the second mixture component. 
For comparison's sake, the same parameters are also used for Scenario B, but with the additional choice of $\alpha_{\bX|1}=\alpha_{\bX|2}=\alpha_{\bY|1}=\alpha_{\bY|2}=0.95$ and $\eta_{\bX|1}=\eta_{\bX|2}=\eta_{\bY|1}=\eta_{\bY|2}=100$.
Two sample sizes are considered: $n=200$ and $n=400$.
Under each scenario, 10,000 replications are considered for each of the two values of $n$; this yields a total of $40,000$ generated data sets.
On each generated data set, both the Gaussian CWM and the contaminated Gaussian CWM are fitted by directly using $k=2$.
The values of the mixture weights in \eqref{eq:generating parameters in group 1} and \eqref{eq:generating parameters in group 2} are chosen to prevent the possible label switching issue (see, e.g., \citealp{Cele:Hurn:Robe:Comp:2000}, \citealp{Step:Deal:2000}, and \citealp{Yao:Mode:2012} for further details about this issue) when the bias and the MSE are computed; the substantial separation between groups helps the algorithms in well-estimating these weights.

The obtained results, in terms of bias and MSE, are summarized in \tablename~\ref{tab:Scenario A} for scenario A, and in \tablename~\ref{tab:Scenario B} for scenario B.
\begin{table}[!ht]
\caption{
Scenario A: estimated biases and MSEs, over 10,000 replications, of the ML estimators of $\bbeta_j$, $j=1,2$, using the Gaussian CWM and the contaminated Gaussian CWM.
\label{tab:Scenario A}
}
\centering
\begin{tabular}{cc c cc c cc}
\toprule
 & & & \multicolumn{2}{c}{Gaussian CWM} & & \multicolumn{2}{c}{Contaminated Gaussian CWM} \\
 & & & $n=200$ & $n=400$ & & $n=200$ & $n=400$ \\
\midrule
Group 1 & Bias & & 
$\begin{pmatrix*}[r]
0.003	& 0.003\\
0.001	& 0.000\\
0.000	& 0.000
\end{pmatrix*}$ 
& 
$\begin{pmatrix*}[r]
0.002	& 0.002\\
0.000	& 0.000\\
0.000	& 0.000
\end{pmatrix*}$ 
&& 
$\begin{pmatrix*}[r]
0.003	& 0.002\\
0.001	& 0.000\\
0.000	& 0.000
\end{pmatrix*}$ 
& 
$\begin{pmatrix*}[r]
0.002	& 0.002\\
0.000	& 0.000\\
0.000	& 0.000
\end{pmatrix*}$
\\[5mm]
& MSE & &  
$\begin{pmatrix*}[r]
0.359	& 0.371\\
0.007	& 0.007\\
0.007 & 0.007
\end{pmatrix*}$ 
& 
$\begin{pmatrix*}[r]
0.177	& 0.177\\
0.004	& 0.004\\
0.003	& 0.003
\end{pmatrix*}$ 
&& 
$\begin{pmatrix*}[r]
0.359	& 0.372\\
0.007	& 0.007\\
0.007	& 0.007
\end{pmatrix*}$ 
& 
$\begin{pmatrix*}[r]
0.177	& 0.177\\
0.004	& 0.004\\
0.003	& 0.003
\end{pmatrix*}$\\[7mm]				
Group 2 & Bias & & 
$\begin{pmatrix*}[r]
-0.001	& -0.001\\
0.000	& 0.000\\
0.001	& 0.000
\end{pmatrix*}$ 
& 
$\begin{pmatrix*}[r]
0.000	& 0.002\\
0.000	& 0.000\\
0.000	& 0.000
\end{pmatrix*}$ 
&& 
$\begin{pmatrix*}[r]
-0.001 & 	-0.001\\
0.000	& 0.000\\
0.001 &  	0.000
\end{pmatrix*}$ 
& 
$\begin{pmatrix*}[r]
0.000	& 0.003\\
0.000	& 0.000\\
0.000	& 0.000
\end{pmatrix*}$
\\[5mm]
        & MSE & &  
$\begin{pmatrix*}[r]
0.148	& 0.148\\
0.003	& 0.003\\
0.003	& 0.003
\end{pmatrix*}$ 
& 
$\begin{pmatrix*}[r]
0.074	& 0.073\\
0.001	& 0.001\\
0.001	& 0.001
\end{pmatrix*}$ 
&& 
$\begin{pmatrix*}[r]
0.148	& 0.148\\
0.003	& 0.003\\
0.003	& 0.003
\end{pmatrix*}$ 
& 
$\begin{pmatrix*}[r]
0.074	& 0.073\\
0.001	& 0.001\\
0.001	& 0.001
\end{pmatrix*}$\\
\bottomrule	
\end{tabular}
\end{table}
\begin{table}[!ht]
\caption{
Scenario B: estimated biases and MSEs, over 10,000 replications, of the ML estimators of $\bbeta_j$, $j=1,2$, using the Gaussian CWM and the contaminated Gaussian CWM.
\label{tab:Scenario B}
}
\centering
\begin{tabular}{cc c cc c cc}
\toprule
 & & & \multicolumn{2}{c}{Gaussian CWM} & & \multicolumn{2}{c}{Contaminated Gaussian CWM} \\
 & & & $n=200$ & $n=400$ & & $n=200$ & $n=400$ \\
\midrule
Group 1 & Bias & & 
$\begin{pmatrix*}[r]
0.000	& 0.001 \\
0.002	& -0.002 \\
-0.002	& 0.002
\end{pmatrix*}$ 
& 
$\begin{pmatrix*}[r]
0.037	& 0.049 \\
0.010	& -0.002 \\
-0.003	& 0.012
\end{pmatrix*}$ 
&& 
$\begin{pmatrix*}[r]
0.002	& -0.007 \\
0.000	& 0.000 \\
0.000	& -0.001
\end{pmatrix*}$ 
& 
$\begin{pmatrix*}[r]
0.000	& 0.002 \\
0.000	& 0.000 \\
0.000	& 0.000
\end{pmatrix*}$
\\[5mm]
& MSE & &  
$\begin{pmatrix*}[r]
1.159	& 1.166 \\
0.021	& 0.020 \\
0.019	& 0.020
\end{pmatrix*}$ 
& 
$\begin{pmatrix*}[r]
0.770	& 0.783 \\
0.014	& 0.011 \\
0.011	& 0.014
\end{pmatrix*}$ 
&& 
$\begin{pmatrix*}[r]
0.250	& 0.239 \\
0.005	& 0.005 \\
0.005	& 0.004
\end{pmatrix*}$ 
& 
$\begin{pmatrix*}[r]
0.067	& 0.071 \\
0.001	& 0.001 \\
0.001	& 0.001
\end{pmatrix*}$\\[7mm]				
Group 2 & Bias & & 
$\begin{pmatrix*}[r]
0.017	& 0.023 \\
0.000	& -0.002 \\
-0.002	& -0.001
\end{pmatrix*}$ 
& 
$\begin{pmatrix*}[r]
0.038	& 0.037 \\
-0.001	& -0.003 \\
-0.004	& -0.001
\end{pmatrix*}$ 
&& 
$\begin{pmatrix*}[r]
-0.005	& 0.001 \\
0.001	& 0.000 \\
0.000	& 0.000
\end{pmatrix*}$ 
& 
$\begin{pmatrix*}[r]
0.002	& 0.000 \\
0.000	& 0.000 \\
0.000	& 0.000
\end{pmatrix*}$
\\[5mm]
        & MSE & &  
$\begin{pmatrix*}[r]
0.258	& 0.259 \\
0.005	& 0.005 \\
0.005	& 0.005
\end{pmatrix*}$ 
& 
$\begin{pmatrix*}[r]
0.123	& 0.128 \\
0.002	& 0.002 \\
0.002	& 0.002
\end{pmatrix*}$ 
&& 
$\begin{pmatrix*}[r]
0.049	& 0.048 \\
0.001	& 0.001 \\
0.001	& 0.001
\end{pmatrix*}$ 
& 
$\begin{pmatrix*}[r]
0.018	& 0.018 \\
0.000	& 0.000 \\
0.000	& 0.000
\end{pmatrix*}$\\
\bottomrule	
\end{tabular}
\end{table}

In all the $40,000$ replications, no convergence problems were observed.
As concerns Scenario A, from \tablename~\ref{tab:Scenario A} it easy to note how the choice of the model has a negligible effect on the estimation of the parameters $\bbeta_1$ and $\bbeta_2$: the biases and MSEs from the two models are practically the same and their values are not substantial (as an example, the maximum obtained absolute value for the bias is 0.003).
These results are not surprising because the generating model is a Gaussian CWM and no local atypical observation is present in any generated data set; in this situation, the contaminated Gaussian CWM tends to the Gaussian CWM.
Finally, for both bias and MSE, it is interesting to note how their values roughly improve with the increase of $n$ and, fixed $n$, with the increase of the size of the considered group (as governed by the values of $\pi_1$ and $\pi_2$).
As concerns Scenario B, the contaminated Gaussian CWM provides estimators of $\bbeta_1$ and $\bbeta_2$ with a lower bias; however, all biases may be considered negligible here.
The very interesting results can be noted in terms of efficiency; here, using the Gaussian CWM instead of the contaminated Gaussian CWM always leads to a substantial increase in the MSE of the estimators of $\bbeta_1$ and $\bbeta_2$.
The increase in the MSE ranges between 314,650\% and 441,146\% when $n = 200$ and between 491,683\% and 1054.952\% when $n=400$.

\subsection{Sensitivity study based on real data}
\label{subsec:Sensitivity study based on real data}

A sensitivity study, based on a real data set, is here described to compare how atypical observations affect the Gaussian CWM and how them are instead handled by the contaminated Gaussian CWM.
The Students data set, introduced by \citet{Ingr:Mino:Punz:Mode:2014} and available at \url{http://www.economia.unict.it/punzo/Data.htm}, is a suitable data set for this purpose.
The data come from a survey of $n=270$ students attending a statistics course at the Department of Economics and Business of the University of Catania in the academic year 2011/2012. 
Although the questionnaire included seven items, the following analysis only concerns, for illustrative purposes, the variables $\mathsf{HEIGHT}$ (height of the respondent, measured in centimeters) and $\mathsf{HEIGHT.F}$ (height of respondent's father, measured in centimeters).
Therefore, the role of $\mathsf{HEIGHT}$ and $\mathsf{HEIGHT.F}$ as response variable and covariate, respectively, is clearly justified.
Moreover, there are $k = 2$ groups of respondents with respect to the gender: 119 males and 151 females.
The scatter plot of the data, with labeling and regression lines based on gender, is shown in \figurename~\ref{fig:realdata}.
\begin{figure}[!ht]
\centering
\subfigure[True labels and regression lines\label{fig:realdata}]
{\resizebox{0.485\textwidth}{!}{\includegraphics{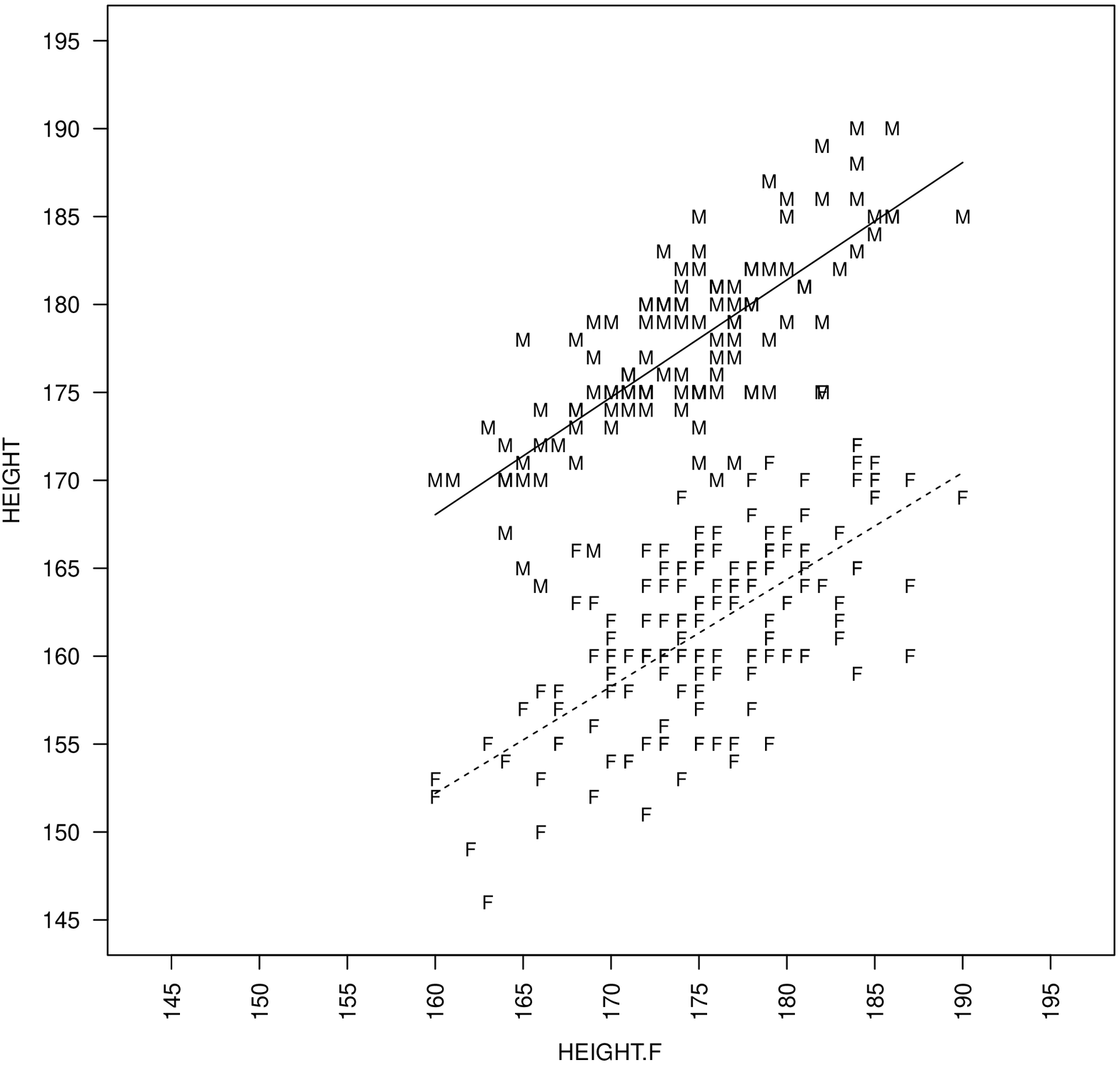}}}
\subfigure[Labels and regression lines from the Gaussian CWM\label{fig:NCWM}]
{\resizebox{0.485\textwidth}{!}{\includegraphics{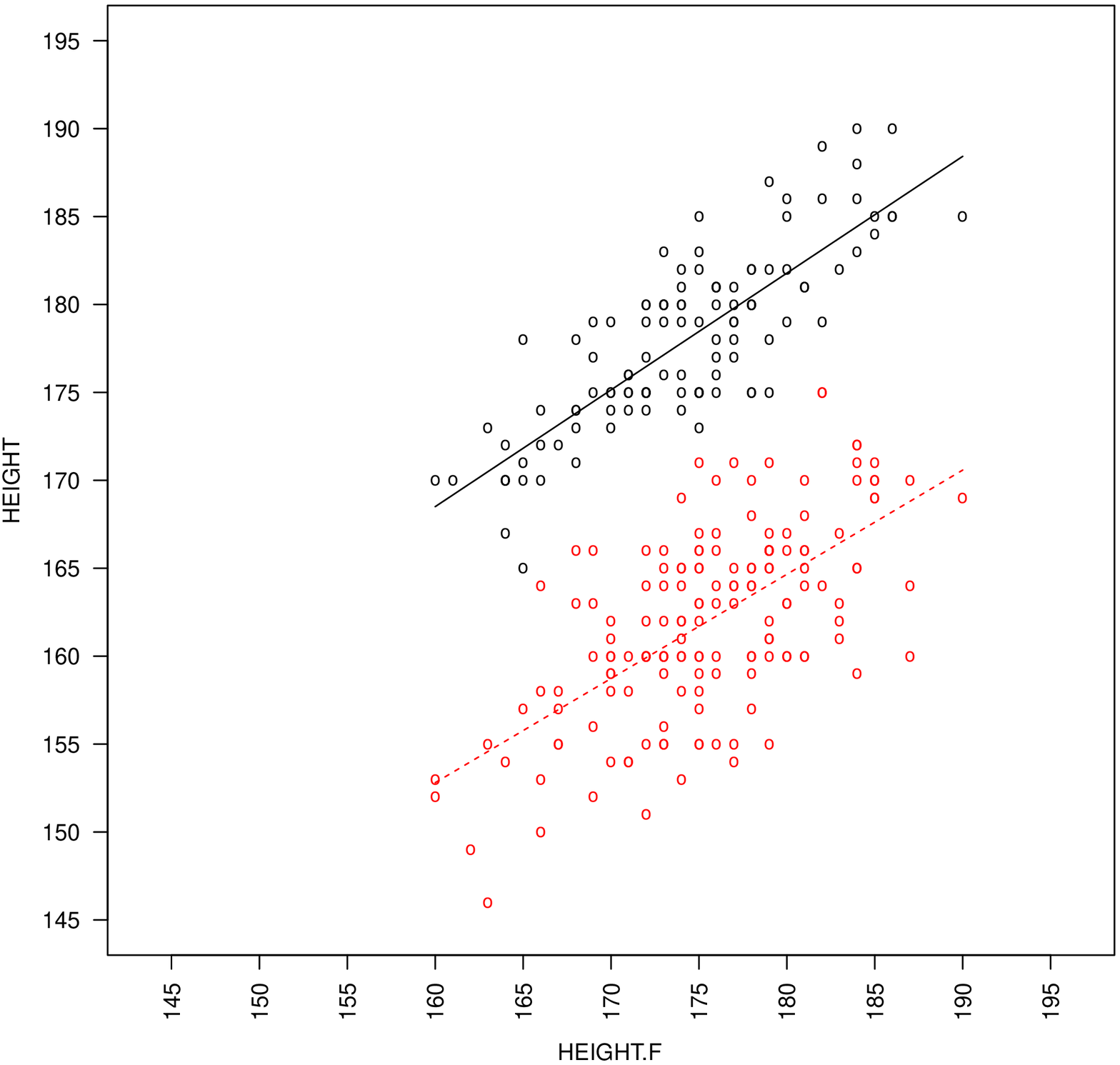}}}
\caption{
Student data: scatter plots and regression lines from the true labeling (on the left; \textsf{M} denotes male and \textsf{F} female) and from the fitting of the Gausian CWM with $k=2$ (on the right).
\label{fig:real data and Gaussian CWM}
}
\end{figure}

By ignoring the classification induced by gender, data are fitted for $k\in\left\{1,2,3\right\}$ according to the Gaussian CWM and the contaminated Gaussian CWM.
\tablename~\ref{tab:BIC values} shows the obtained BIC values.
\begin{table}[!ht]
\caption{
BIC values on the original data. 
}
\label{tab:BIC values}
\centering
\begin{tabular}{lrr}
\toprule
     	   &	 Gaussian CWM	 &	contaminated Gaussian CWM   \\
\midrule															
$k=1$  & -3710.469 & -3732.909 \\
$k=2$  & -3601.953 & -3646.741 \\
$k=3$  & -3767.954 & -3835.135 \\
\bottomrule	
\end{tabular}
\end{table}
The best model is the Gaussian CWM with $k=2$; the corresponding classification and regression lines are displayed in \figurename~\ref{fig:NCWM}.
Based on \figurename~\ref{fig:realdata}, the estimated regression lines appear to be in agreement with the true ones.
The classification is good too: the model only yields six misclassified observations (six males erroneously considered as females), corresponding to a very low misclassification rate of 0.022. 
This model will be considered as the benchmark to judge the results of the next two sections.

%
%

\subsubsection{Adding a single atypical point}
\label{subsubsec:Adding a single atypical point}

The first sensitivity analysis aims to evaluate the impact of a single atypical observation on the fitting of the local regression lines for the Gaussian CWM and the contaminated Gaussian CWM.
With this end, fifteen ``perturbed'' data sets are generated by adding an atypical point to the data.
These points are all displayed together, as bullets, in \figurename~\ref{fig:atypicaldata}.
They represent different types of local atypical observations in accordance to \tablename~\ref{tab:Atypical observation labelling}.
\begin{figure}[!ht]
\centering
  \resizebox{0.49\textwidth}{!}{
	\includegraphics{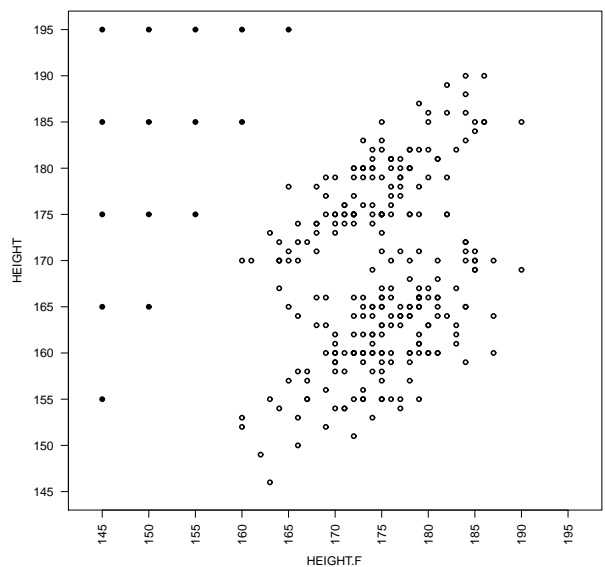}
	}%
\caption{
Student data: scatter plot with each \small{$\bullet$} denoting the observation added to perturb the original data.
}
\label{fig:atypicaldata}       
\end{figure} 

For each perturbed data set, the Gaussian CWM and the contaminated Gaussian CWM are fitted with $k=2$.
In all the fifteen scenarios, the contaminated Gaussian CWM detects only one atypical observation, the true one.
Moreover, while the regression lines from the contaminated Gaussian CWM are not substantially different from those displayed in \figurename~\ref{fig:NCWM}, there are some scenarios where one of the regression lines from the Gaussian CWM is severely dragged towards the atypical point.
This happens for the atypical points on the top-left corner of \figurename~\ref{fig:atypicaldata}; the most representative example is given in \figurename~\ref{fig:worst situation} (the entire set of plots is not reported here for brevity's sake).  
\begin{figure}[!ht]
\centering
\subfigure[Gaussian CWM \label{fig:NCWM145e195}]
{\resizebox{0.49\textwidth}{!}{\includegraphics{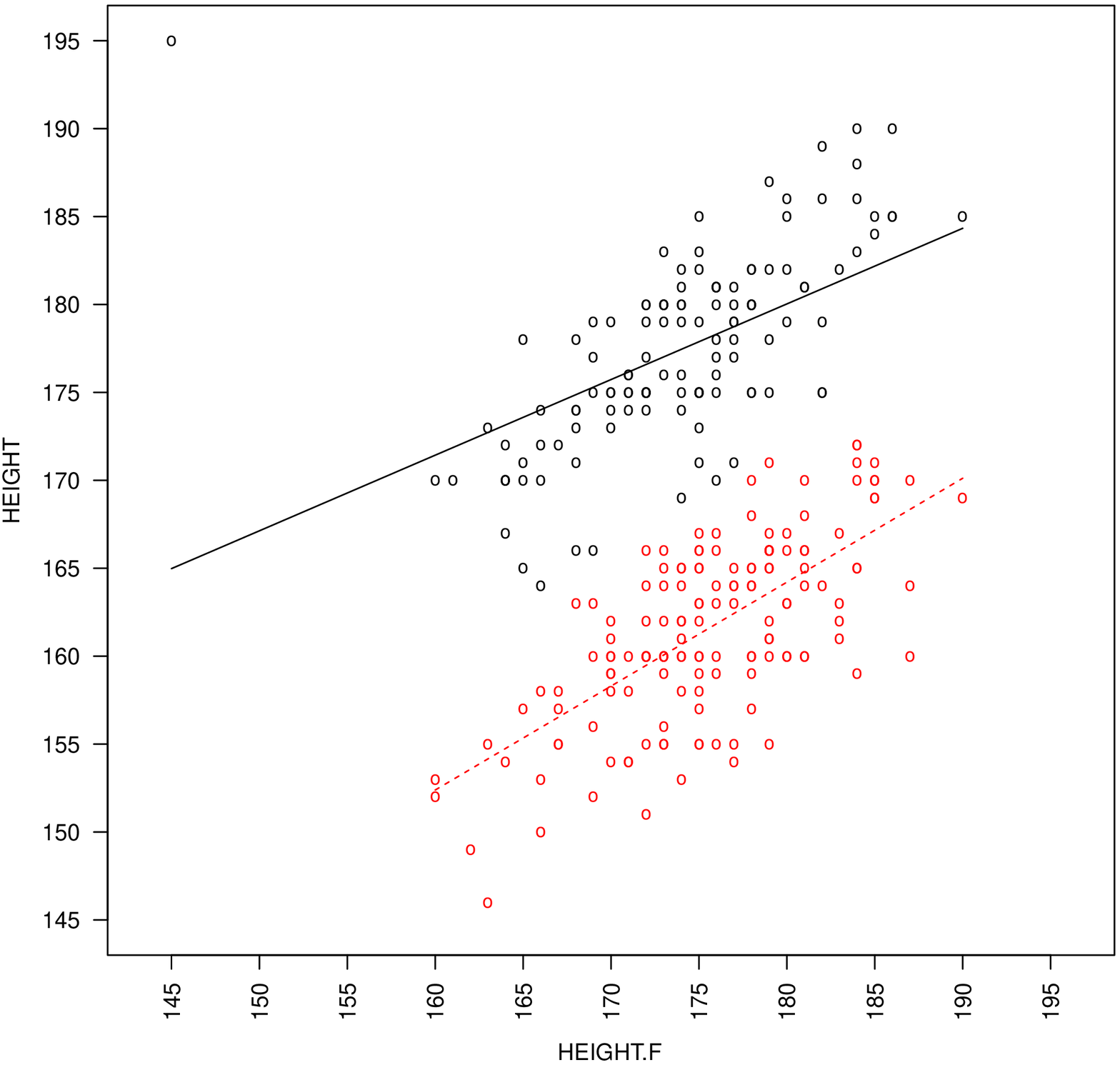}}}
\subfigure[Contaminated Gaussian CWM \label{fig:CNCWM145e195}]
{\resizebox{0.49\textwidth}{!}{\includegraphics{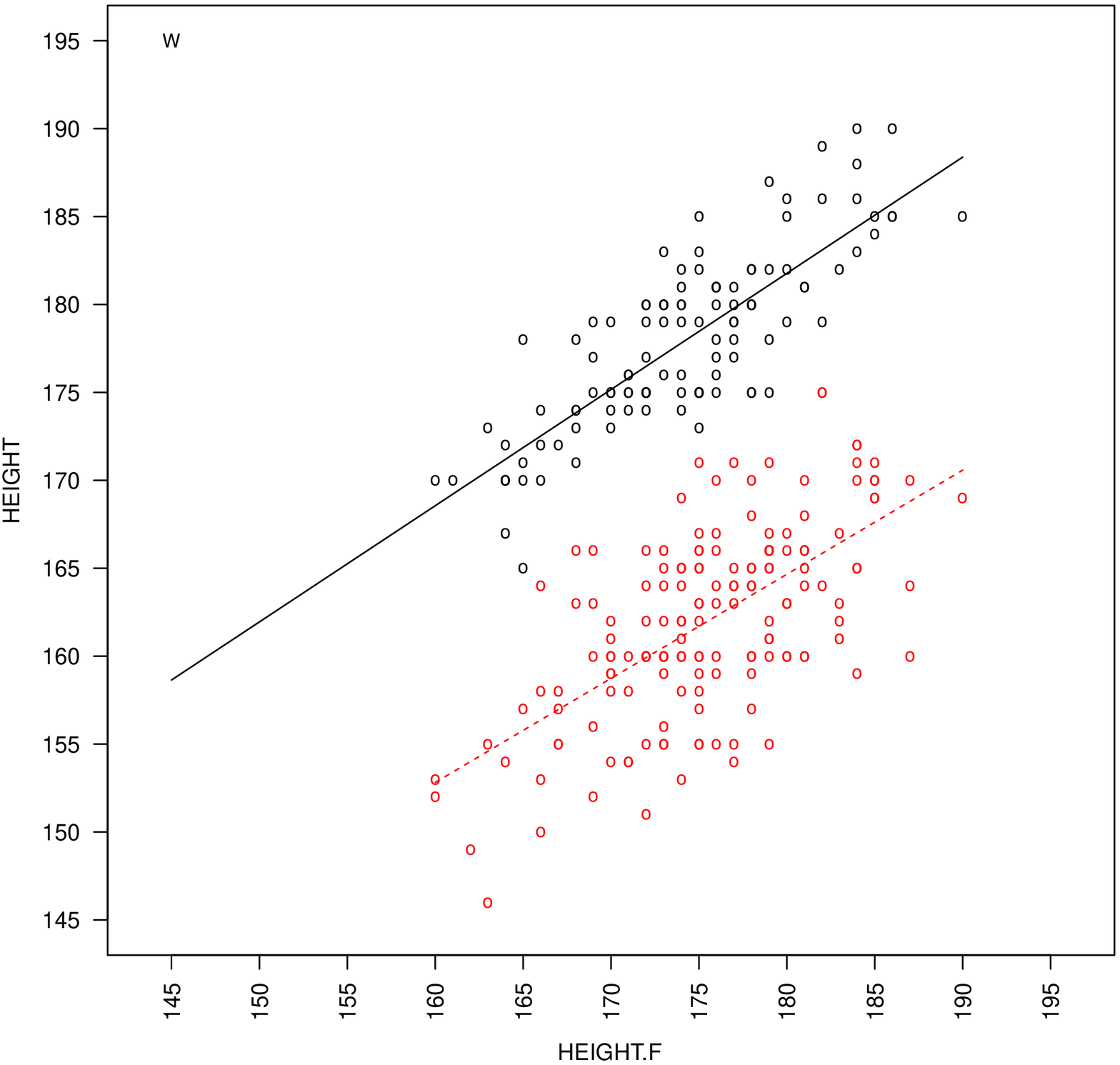}}}
\caption{
Perturbed Student data: scatter plots, labels, and regression lines from the Gaussian CWM (on the left) and from the contaminated Gaussian CWM (on the right; \textsf{W} indicates a detected local bad leverage point).
\label{fig:worst situation}
}
\end{figure}
In \figurename~\ref{fig:CNCWM145e195}, the label ``\textsf{W}'' indicates that the contaminated Gaussian CWM, based on the rule given in \tablename~\ref{tab:practical detection}, detects that point as atypical both on \textsf{HEIGHT.F} and $\mathsf{HEIGHT}|\mathsf{HEIGHT.F}$; in other words, this observation is a local bad leverage point according to \tablename~\ref{tab:Atypical observation labelling}.
On the contrary, \figurename~\ref{fig:best situation} shows a scenario where the two models provide similar results.
In \figurename~\ref{fig:CNCWM145e155}, the label ``\textsf{X}'' indicates that the contaminated Gaussian CWM detects that point as locally atypical only on \textsf{HEIGHT.F}; it is a good leverage point according to \tablename~\ref{tab:Atypical observation labelling}.  
\begin{figure}[!ht]
\centering
\subfigure[Gaussian CWM \label{fig:NCWM145e155}]
{\resizebox{0.49\textwidth}{!}{\includegraphics{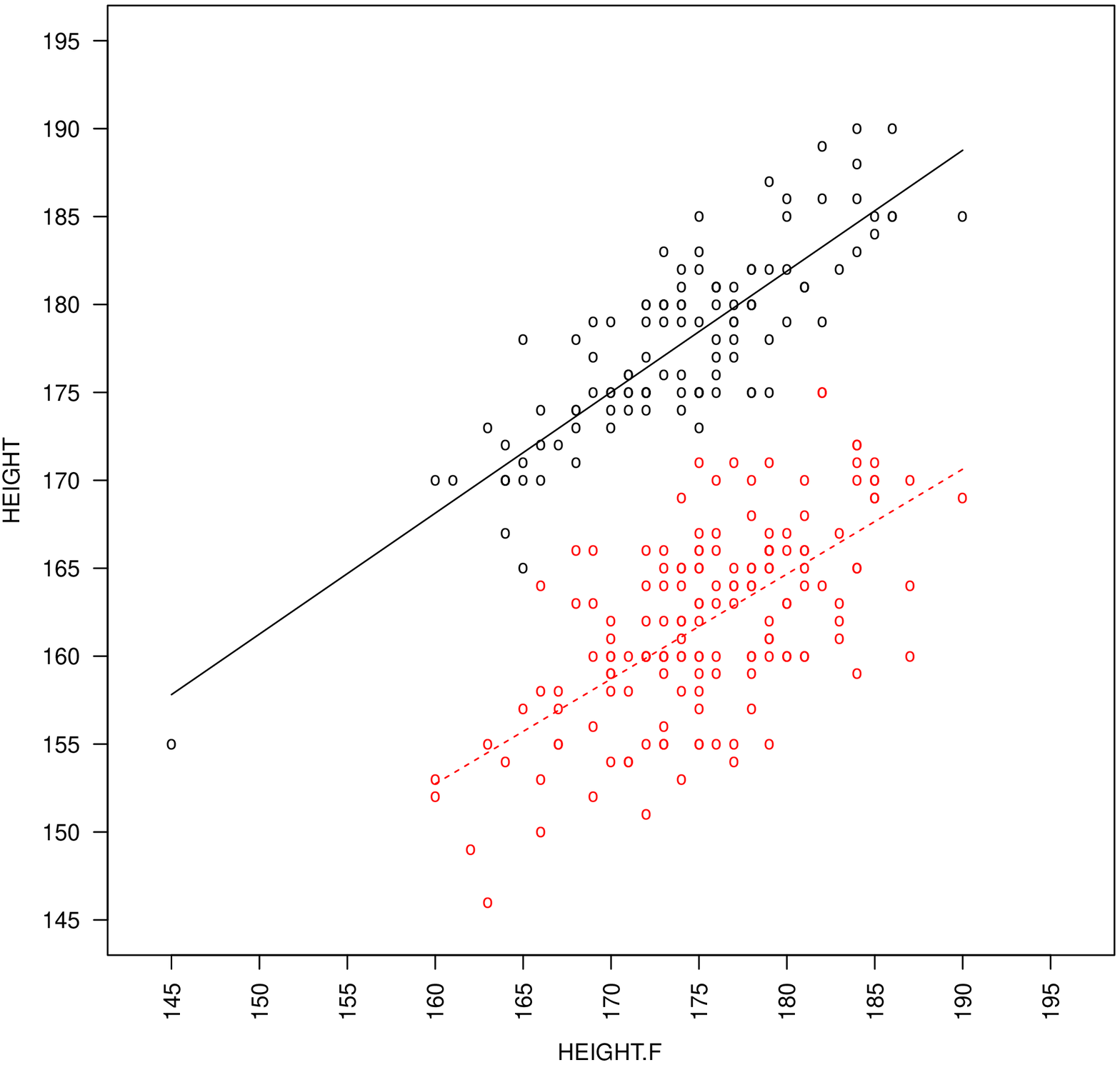}}}
\subfigure[Contaminated Gaussian CWM \label{fig:CNCWM145e155}]
{\resizebox{0.49\textwidth}{!}{\includegraphics{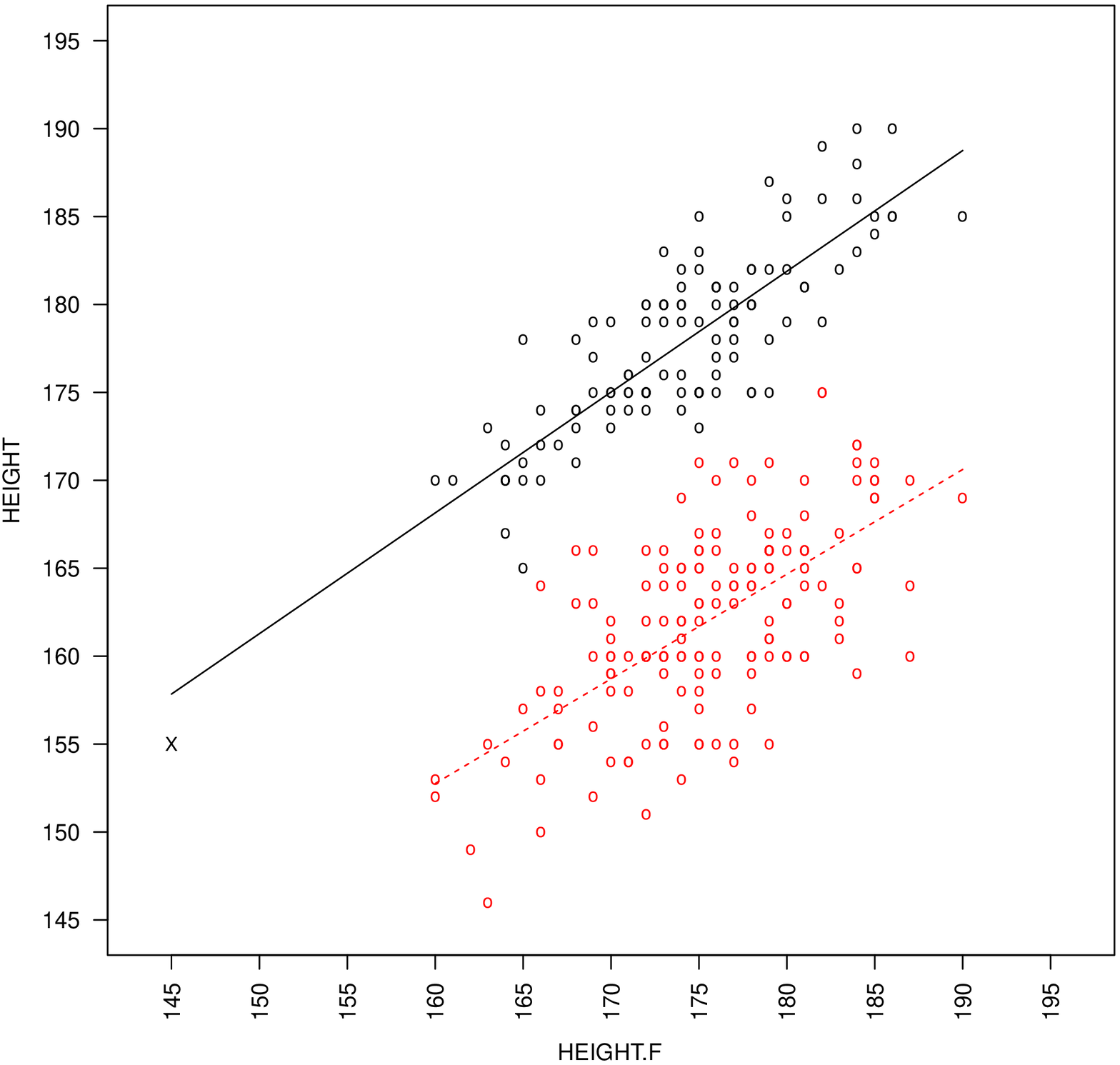}}}
\caption{
Perturbed Student data: scatter plots, labels, and regression lines from the Gaussian CWM (on the left) and from the contaminated Gaussian CWM (on the right; \textsf{X} indicates a detected local good leverage point).
\label{fig:best situation}
}
\end{figure}

To complete the analysis for the contaminated Gaussian CWM, \tablename~\ref{tab:eta values} shows the estimated values of the degrees of contamination $\eta_{\mathsf{HEIGHT.F}}$ and $\eta_{\mathsf{HEIGHT}}$ in the group containing the atypical point.
\begin{table}[!b]
\caption{
Values of the degrees of contamination $\left(\eta_{\mathsf{HEIGHT.F}},\eta_{\mathsf{HEIGHT}}\right)$, for each perturbed data set, in the component of the contaminated Gaussian CWM containing the atypical observation. 
}
\label{tab:eta values}
\centering
\begin{tabular}{l| ccccc}
\toprule
 & \multicolumn{5}{l}{\textsf{HEIGHT.F}} \\    	 
 \textsf{HEIGHT} & 145 & 150 & 155 & 160 & 165 \\
\midrule															
195   &  $\left(10.083,93.845\right)$     &  $\left(4.637,74,335\right)$     &  $\left(1.529,57.424\right)$    & $\left(1.019,43.559\right)$ & $\left(1.005,32.226\right)$ \\
185   &  $\left(10.087,43.225\right)$     &  $\left(4.632,31.968\right)$     &  $\left(1.502,22.836\right)$    & $\left(1.004,15.125\right)$ &    \\
175   &  $\left(10.070,14.562\right)$     &  $\left(4.650,\ \ 7.119\right)$  &  $\left(1.654,\ \ 2.419\right)$ &    &    \\
165   &  $\left(10.132,\ \ 1.086\right)$  &  $\left(4.652,\ \ 1.000\right)$  &     &    &    \\
155   &  $\left(10.073,\ \ 1.089\right)$  &      &     &    &    \\
\bottomrule	
\end{tabular}
\end{table} 
As expected, the estimate of $\eta_{\mathsf{HEIGHT.F}}$ increases as the value of \textsf{HEIGHT.F}, for the atypical point, further departs from the bulk of the values of \textsf{HEIGHT.F} in its group of membership, regardless from the value of \textsf{HEIGHT}; 
this can be easily noted by looking at \tablename~\ref{tab:eta values} column-by-column.   
Finally, the estimate of $\eta_{\mathsf{HEIGHT}}$ increases as the atypical point further departs from the regression line of the group the atypical point is assigned.

\subsubsection{Adding uniform noise}
\label{subsubsec:Adding uniform noise}

The second sensitivity analysis aims to evaluate the impact of noise on fitting and clustering from the Gaussian CWM and the contaminated Gaussian CWM.
With this end, we modify the original data by including twenty noisy points generated from a uniform distribution over a square centered on the bivariate mean $\left(174.963,168.652\right)$ of the observations and with side of length $60$ (centimeters).
This square contains the original data. 
\figurename~\ref{fig:datawithnoise} shows the modified data set with bullets denoting uniform noise points.
\begin{figure}[!ht]
\centering
  \resizebox{0.49\textwidth}{!}{
	\includegraphics{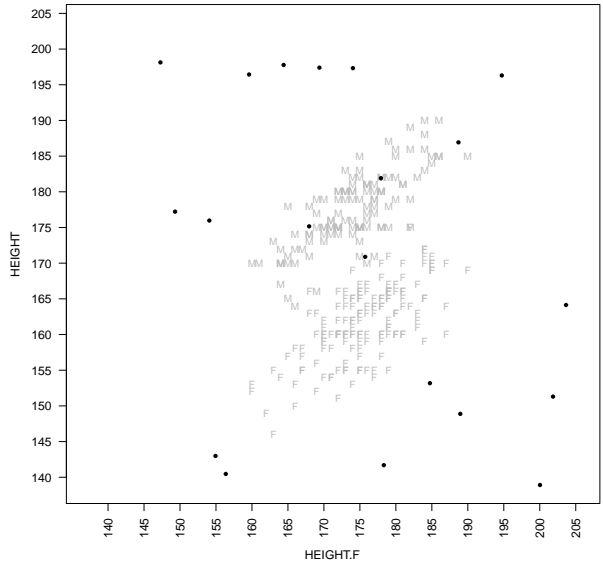}
	}%
\caption{
Scatter plot of the Student data where the added uniform noise points are denoted by~$\bullet$.
}
\label{fig:datawithnoise}       
\end{figure}

\tablename~\ref{tab:noise - BIC values} shows the BIC values, in correspondence of $k\in\left\{ 1,2,3\right\}$, for the Gaussian CWM and the contaminated Gaussian CWM. 
\begin{table}[!ht]
\caption{
BIC values on the data with uniform noise. 
}
\label{tab:noise - BIC values}
\centering
\begin{tabular}{lrr}
\toprule
     	   &	 Gaussian CWM	 &	contaminated Gaussian CWM   \\
\midrule															
$k=1$  & -4206.043 & -4228.757 \\
$k=2$  & -4178.727 & -4142.435 \\
$k=3$  & -4209.625 & -4152.590 \\
\bottomrule	
\end{tabular}
\end{table}
Generally, the best model is the contaminated Gaussian CWM with $k=2$.
Among the fitted Gaussian CWMs, the best one, in terms of BIC, has $k=2$ components.
For comparison's sake, these models are displayed in \figurename~\ref{fig:CWMs with noise}.
\begin{figure}[!ht]
\centering
\subfigure[Gaussian CWM \label{fig:NCWMnoise}]
{\resizebox{0.49\textwidth}{!}{\includegraphics{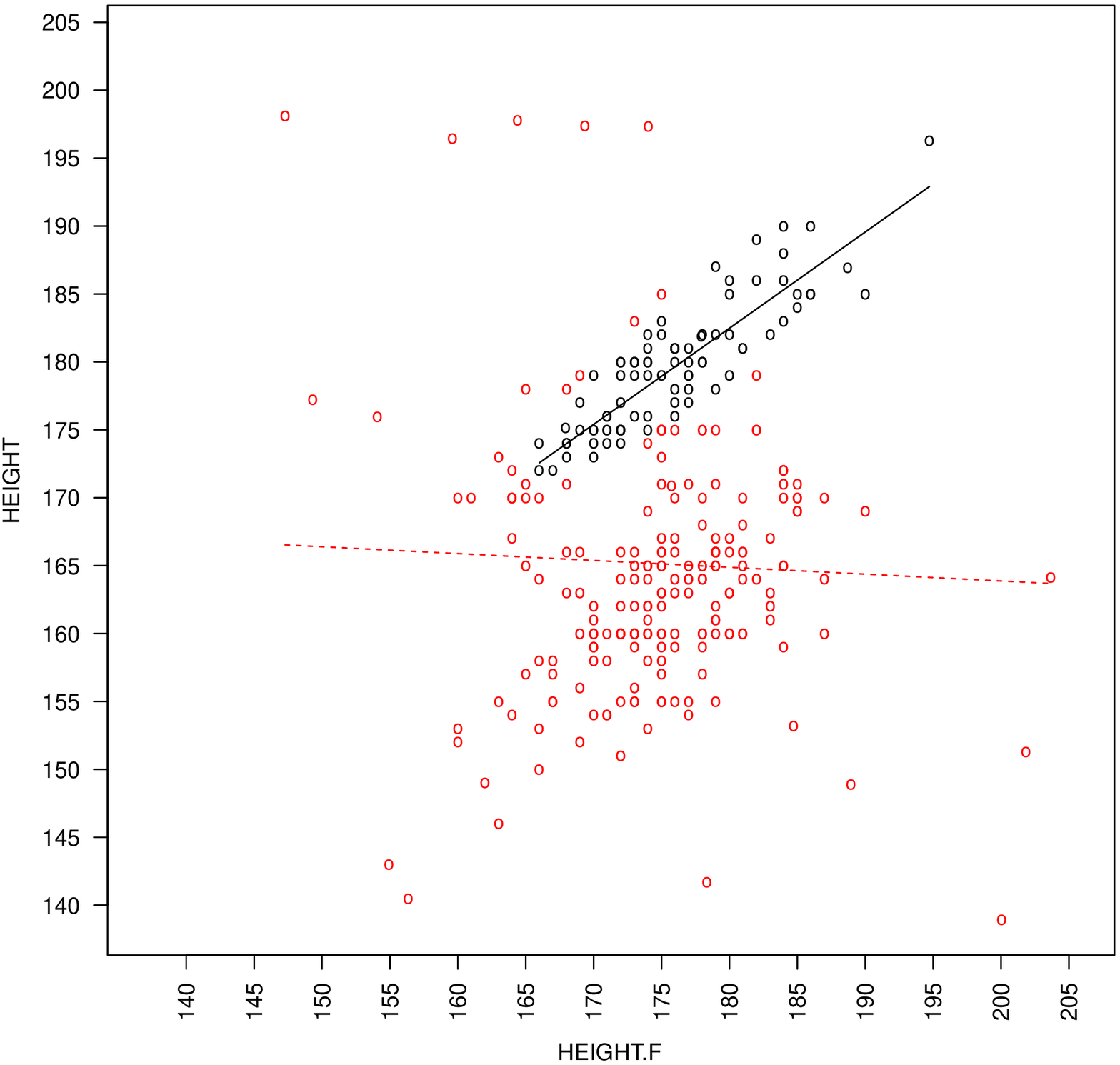}}}
\subfigure[Contaminated Gaussian CWM \label{fig:CNCWMnoise}]
{\resizebox{0.49\textwidth}{!}{\includegraphics{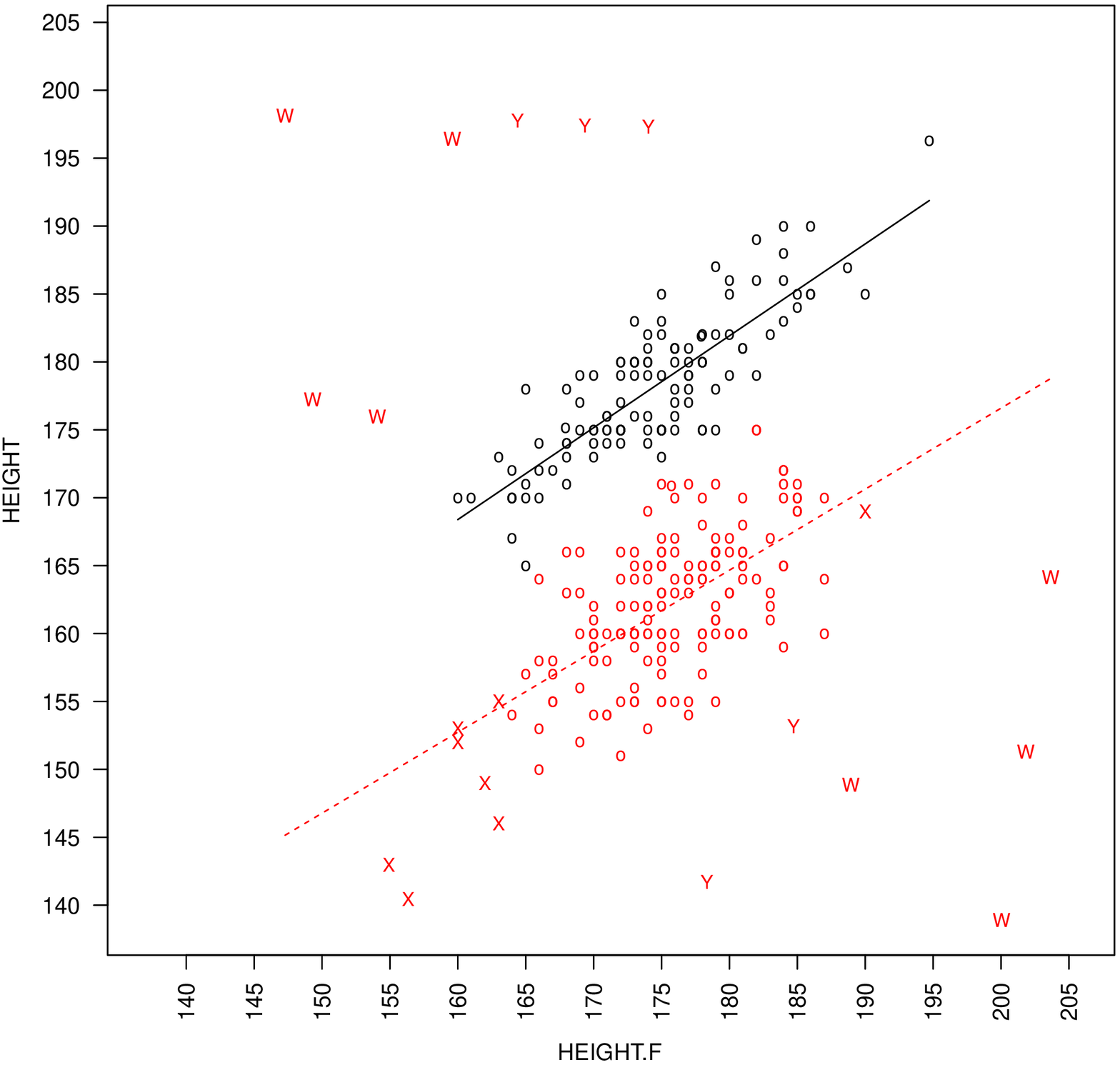}}}
\caption{
Student data with noise: scatter plots, labels, and regression lines from the Gaussian CWM (on the left) and from the contaminated Gaussian CWM (on the right; \textsf{W} denotes bad leverage points, \textsf{X} denotes good leverage points, and \textsf{Y} denotes outliers).
\label{fig:CWMs with noise}
}
\end{figure}
It is important to note that, in \figurename~\ref{fig:CNCWMnoise}, \textsf{Y} denotes the detected outliers, \textsf{X} indicates the detected good leverage points, and \textsf{W} denotes the bad leverage points.
Still importantly, we can see the poor results obtained by the Gaussian CWM, where the regression line referred to the females is severely affected by the noisy observations; as a by-product in clustering terms, the number of original observations misclassified increases from 6 --- obtained by the model on the original data only --- to 36.
On the contrary, our model maintains at 6 the number of misclassified original observations and provides regression lines which are in agreement with those displayed in \figurename~\ref{fig:NCWM}. 
Finally, our model is able to classify each observation, with respect to its group of membership, in accordance to the four categories given in \tablename~\ref{tab:Atypical observation labelling}.




\section{Discussion}
\label{sec:Discussion and future work}

The contaminated Gaussian CWM has been introduced as a generalization of the Gaussian CWM \citep{Dang:Punz:McNi:Ingr:Brow:Mult:2014} that accommodates atypical observations; the analyses of Section~\ref{sec:Numerical studies} have shown its usefulness. 
More importantly, however, the contaminated Gaussian CWM is put forward as a gold standard for robust clustering in regression analysis, where observations, in addition to be assigned to the groups, also need to be classified in one of the four categories given in \tablename~\ref{tab:Atypical observation labelling}.  
Although approaches such as mixtures of $t$ regression models, mixtures of Laplace regression models, and $t$ CWMs, can be used for robust clustering in regression analysis, they assimilate atypical points into clusters rather than separating them out in a direct way.
Clusterwise linear regression through trimming can also be used, but it requires \textit{a~priori} specification of the proportion of outliers and leverage points, but this is not always possible in practice; in fact, it is all but impossible if the data cannot easily be visualized.

Another distinct advantage of our contaminated Gaussian CWM over clusterwise linear regression through trimming is that we can easily extend the approach to model-based classification \citep[see, e.g.,][]{McNi:Mode:2010} and model-based discriminant analysis \citep{Hast:Tibs:Disc:1996}. 
In fact, there are a number of options for the type of supervision that could be used in partial classification applications for our model, i.e., one could specify some of the $\bz_i$ and/or some of the $\bu_i$ and $\bv_i$ \textit{a~priori}, $i=1,\ldots,n$. 
This provides yet more flexibility than exhibited by any competing approach, as does the ability of our approach to work in higher dimensions where atypical observations cannot easily be visualized.


Future work will focus on the following avenues.
\begin{itemize}
	\item The development of an {\sf R} package to facilitate dissemination of our contaminated Gaussian CWM.
	\item It would be interesting to investigate the sample breakdown points for the proposed method. 
However, we should note that the analysis of breakdown point for traditional linear regression cannot be directly applied to mixtures of regression models. 
\citet{Garc:Gord:Mayo:SanM:Robu:2010} also stated that the traditional definition of breakdown point is not the right one to quantify the robustness of mixtures of regression models to atypical observations, since the robustness of these procedures is not only data dependent but also cluster dependent.
\citet{Henn:Break:2004} provided a new definition of breakdown points for mixture models based on the breakdown of at least one of the mixture components.
Based on the results of \citet{Henn:Break:2004} about mixtures of $t$ distributions, we guess that only extreme outliers would lead to the breakdown of the contaminated Gaussian CWM.
Therefore, we believe that the model can still be used as a robust approach with the exception of extreme atypical observations that, however, can easily be deleted. 
	\item In the fashion of \citet{Banf:Raft:mode:1993} and \citet{Cele:Gova:Gaus:1995}, and more directly according to \citet{Punz:McNi:Robu:2013} and \citet{Dang:Punz:McNi:Ingr:Brow:Mult:2014}, the proposed approach could be made more flexible and parsimonious by imposing constraints on the eigen-decomposed component matrices $\bSigma_{\bX|j}$ and $\bSigma_{\bY|j}$, $j=1,\ldots,k$. 
	In the fashion of \citet{Sube:Punz:Ingr:McNi:Clus:2013,Sube:Punz:Ingr:McNi:Clus:2014} and \citet{Punz:McNi:Robu:2014}, parsimony, but also dimension reduction, could be obtained by exploiting local factor analyzers.
	\item Further developments of our model could be obtained by studying the asymptotic properties of the ML estimators and by defining statistical tests for evaluating the significance of the regression coefficients $\bbeta_j$, $j=1,\ldots,k$. 
	Moreover, still working on the eigen-decomposed component matrices $\bSigma_{\bX|j}$ and $\bSigma_{\bY|j}$, $j=1,\ldots,k$, in the fashion of \citet{Ingr:Alik:2004}, \citet{Ingr:Rocc:Cons:2007}, and \citet{Brow:Sube:McNi:Cons:2013}, suitable constraints on their eigenvalues during the ECM algorithm could attenuate possible problems on the likelihood function such as unboundedness and spurious local maxima \citep[see also][]{Seo:Kim:Root:2012}. 
\end{itemize}


\appendix

\section{Proof of Proposition~\ref{pro:contaminated Gaussian CWM versus mixtures of contaminated Gaussian regressions}}
\label{app:Proof of Proposition 1}

\begin{proof}
Under the assumptions of the proposition, model~\eqref{eq:conditional from a CN CWM} simplifies as
\begin{align*}
p\left(\by|\bx;\bvartheta\right)=&
\sum_{j=1}^k
\frac{\pi_jf\left(\bx;\bmu_{\bX},\bSigma_{\bX},\alpha_{\bX},\eta_{\bX}\right)}{\displaystyle\sum_{h=1}^k\pi_hf\left(\bx;\bmu_{\bX},\bSigma_{\bX},\alpha_{\bX},\eta_{\bX}\right)}f\left(\by;\bmu_{\bY}\left(\bx;\bbeta_j\right),\bSigma_{\bY|j},\alpha_{\bY|j},\eta_{\bY|j}\right)\nonumber\\
=&\sum_{j=1}^k
\pi_jf\left(\by;\bmu_{\bY}\left(\bx;\bbeta_j\right),\bSigma_{\bY|j},\alpha_{\bY|j},\eta_{\bY|j}\right),
\end{align*}
which corresponds to the conditional distribution from a mixture of contaminated Gaussian regression models as defined by \eqref{eq:mixtures of contaminated Gaussian regressions}.
\end{proof}

\section{Proof of Proposition~\ref{pro:1}}
\label{app:Proof of Proposition 2}

\begin{proof}
Suppose that 
\begin{equation}
p\left(\bx,\by;\bvartheta\right)=p\left(\bx,\by;\widetilde{\bvartheta}\right).
\label{eq:equality of parameterizations}
\end{equation}
Integrating out $\by$ from \eqref{eq:equality of parameterizations} yields 
\begin{equation}
\sum_{j=1}^k\pi_jf\left(\bx;\bmu_{\bX|j},\bSigma_{\bX|j},\alpha_{\bX|j},\eta_{\bX|j}\right)\nonumber\\
=
\sum_{s=1}^{\widetilde{k}}\widetilde{\pi}_sf\left(\bx;\widetilde{\bmu}_{\bX|s},\widetilde{\bSigma}_{\bX|s},\widetilde{\alpha}_{\bX|s},\widetilde{\eta}_{\bX|s}\right),
\label{eq:equality of parameterizations only on the covariates}
\end{equation}
corresponding to the marginal distribution of $\bX$, say $p\left(\bx;\bpi,\bmu_{\bX},\bSigma_{\bX},\balpha_{\bX},\boldsymbol{\eta}_{\bX}\right)$.
Dividing \eqref{eq:equality of parameterizations} by the left-hand side of \eqref{eq:equality of parameterizations only on the covariates} leads to
\begin{align}
p\left(\by|\bx;\bvartheta\right)&=\sum_{j=1}^k
\frac{\pi_j
f\left(\bx;\bmu_{\bX|j},\bSigma_{\bX|j},\alpha_{\bX|j},\eta_{\bX|j}\right)}{\displaystyle\sum_{t=1}^k\pi_tf\left(\bx;\bmu_{\bX|t},\bSigma_{\bX|t},\alpha_{\bX|t},\eta_{\bX|t}\right)}
f\left(\by;\bmu_{\bY}\left(\bx;\bbeta_j\right),\bSigma_{\bY|j},\alpha_{\bY|j},\eta_{\bY|j}\right)\nonumber\\
&= \sum_{s=1}^{\widetilde{k}}
\frac{\widetilde{\pi}_s
f\left(\bx;\widetilde{\bmu}_{\bX|s},\widetilde{\bSigma}_{\bX|s},\widetilde{\alpha}_{\bX|s},\widetilde{\eta}_{\bX|s}\right)}{\displaystyle\sum_{t=1}^k\pi_tf\left(\bx;\bmu_{\bX|t},\bSigma_{\bX|t},\alpha_{\bX|t},\eta_{\bX|t}\right)}
f\left(\by;\bmu_{\bY}\left(\bx;\widetilde{\bbeta}_s\right),\widetilde{\bSigma}_{\bY|s},\widetilde{\alpha}_{\bY|s},\widetilde{\eta}_{\bY|s}\right)=p\left(\by|\bx;\widetilde{\bvartheta}\right).
\label{eq:equality of parameterizations on the response}
\end{align}
For each fixed $\bx$, this is a mixture of contaminated Gaussian distributions for $\bY$ \citep[cf.][]{Punz:McNi:Robu:2013}.

Following the scheme of \citet[][p.~292]{Henn:Iden:2000}, define the set of all covariate points $\bx$ which can be used to distinct different regression coefficients $\bbeta_j$ by different values of $\bmu_{\bY}\left(\bx;\bbeta_j\right)$:
\begin{align*}
	\mathcal{X}=\Big\{\bx\in\real^{d_{\bX}}:\ &\forall\ j,l\in \left\{1,\ldots,k\right\} \text{ and } s,t\in \left\{1,\ldots,\widetilde{k}\right\}, \\
	&\bmu_{\bY}\left(\bx;\bbeta_j\right)=\bmu_{\bY}\left(\bx;\bbeta_l\right)\ \Rightarrow\ \bbeta_j=\bbeta_l,\\ 
	&\bmu_{\bY}\left(\bx;\bbeta_j\right)=\bmu_{\bY}\left(\bx;\widetilde{\bbeta}_s\right)\ \Rightarrow\ \bbeta_j=\widetilde{\bbeta}_s,\\ 
	&\bmu_{\bY}\left(\bx;\widetilde{\bbeta}_s\right)=\bmu_{\bY}\left(\bx;\widetilde{\bbeta}_t\right)\ \Rightarrow\ \widetilde{\bbeta}_s=\widetilde{\bbeta}_t
	\Big\}.
\end{align*}
Note that, $\mathcal{X}$ is complement of a finite union of hyperplanes of $\real^{d_{\bX}}$.
Therefore, 
\begin{displaymath}
\int_{\mathcal{X}}p\left(\bx;\bpi,\bmu_{\bX},\bSigma_{\bX},\balpha_{\bX},\boldsymbol{\eta}_{\bX}\right)d\bx=1.	
\end{displaymath}
For $\bx\in\mathcal{X}$, all pairs $\left(\bmu_{\bY}\left(\bx;\bbeta_j\right),\bSigma_{\bY|j}\right)$, $j=1,\ldots,k$, are pairwise distinct, because all $\left(\bbeta_j,\bSigma_{\bY|j}\right)$, $j=1,\ldots,k$, are pairwise distinct for the condition~\eqref{eq:sufficient condition} of the proposition.
Based on \cite{Punz:McNi:Robu:2013}, such a condition also guarantees that, for each $\bx\in\mathcal{X}$, model~\eqref{eq:equality of parameterizations on the response} is identifiable and this implies that $k=\widetilde{k}$ and also implies that for each $j\in\left\{1,\ldots,k\right\}$ there exists an $s\in\left\{1,\ldots,k\right\}$ such that 
\begin{displaymath}
\alpha_{\bY|j}=\widetilde{\alpha}_{\bY|s},
\quad \bbeta_j=\widetilde{\bbeta}_s,
\quad \bSigma_{\bY|j}=\widetilde{\bSigma}_{\bY|s},
\quad \eta_{\bY|j}=\widetilde{\eta}_{\bY|s}
\end{displaymath}
and 
\begin{equation}
\frac{\pi_j
f\left(\bx;\bmu_{\bX|j},\bSigma_{\bX|j},\alpha_{\bX|j},\eta_{\bX|j}\right)}{\displaystyle\sum_{t=1}^k\pi_tf\left(\bx;\bmu_{\bX|t},\bSigma_{\bX|t},\alpha_{\bX|t},\eta_{\bX|t}\right)}=\frac{\widetilde{\pi}_s
f\left(\bx;\widetilde{\bmu}_{\bX|s},\widetilde{\bSigma}_{\bX|s},\widetilde{\alpha}_{\bX|s},\widetilde{\eta}_{\bX|s}\right)}{\displaystyle\sum_{t=1}^k\pi_tf\left(\bx;\bmu_{\bX|t},\bSigma_{\bX|t},\alpha_{\bX|t},\eta_{\bX|t}\right)}. 
\label{eq:contaminated weights 1}
\end{equation}
Now, based on \eqref{eq:equality of parameterizations}, the equality in \eqref{eq:contaminated weights 1} simplifies as
\begin{equation}
\pi_j
f\left(\bx;\bmu_{\bX|j},\bSigma_{\bX|j},\alpha_{\bX|j},\eta_{\bX|j}\right)=\widetilde{\pi}_s
f\left(\bx;\widetilde{\bmu}_{\bX|s},\widetilde{\bSigma}_{\bX|s},\widetilde{\alpha}_{\bX|s},\widetilde{\eta}_{\bX|s}\right),\quad \forall\ \bx\in\mathcal{X}. 
\label{eq:contaminated weights 2}
\end{equation}
Integrating \eqref{eq:contaminated weights 2} over $\bx\in\mathcal{X}$ yields $\pi_j=\widetilde{\pi}_s$.
Therefore, condition~\eqref{eq:contaminated weights 2} further simplifies as
\begin{equation*}
f\left(\bx;\bmu_{\bX|j},\bSigma_{\bX|j},\alpha_{\bX|j},\eta_{\bX|j}\right)=
f\left(\bx;\widetilde{\bmu}_{\bX|s},\widetilde{\bSigma}_{\bX|s},\widetilde{\alpha}_{\bX|s},\widetilde{\eta}_{\bX|s}\right),\quad \forall\ \bx\in\mathcal{X}. 
\end{equation*}
The equalities $\alpha_{\bX|j}=\widetilde{\alpha}_{\bX|s}$, $\bmu_{\bX|j}=\widetilde{\bmu}_{\bX|s}$, $\bSigma_{\bX|j}=\widetilde{\bSigma}_{\bX|s}$, and $\eta_{\bX|j}=\widetilde{\eta}_{\bX|s}$ simply arise from the identifiability of the contaminated Gaussian distribution, and this completes the proof.
\end{proof}

\section{Updates in the first CM-step}
\label{app:Updates}

The estimates of $\pi_j$, $\mu_{\bX|j}$, $\bSigma_{\bX|j}$, $\alpha_{\bX|j}$, $\bbeta_j$, $\bSigma_{\bY|j}$, and $\alpha_{\bY|j}$, $j=1,\ldots,k$, at the $\left(r+1\right)$th first CM-step of the ECM algorithm, require the maximization of 
\begin{equation}
Q\left(\bvartheta_1|\bvartheta^{\left(r\right)}\right)
=Q_1\left(\bpi|\bvartheta^{\left(r\right)}\right)
+Q_2\left(\balpha_{\bX}|\bvartheta^{\left(r\right)}\right)
+Q_3\left(\bmu_{\bX},\bSigma_{\bX}|\bvartheta^{\left(r\right)}\right)
+Q_4\left(\balpha_{\bY}|\bvartheta^{\left(r\right)}\right)
+Q_5\left(\bbeta,\bSigma_{\bY}|\bvartheta^{\left(r\right)}\right),
\label{eq:global Q}
\end{equation}
where  
\begin{align*}
&Q_1\left(\bpi|\bvartheta^{\left(r\right)}\right)=\sum_{i=1}^{n}\sum_{j=1}^{k}z_{ij}^{\left(r\right)}\ln \pi_j,\qquad\qquad
Q_2\left(\balpha_{\bX}|\bvartheta^{\left(r\right)}\right)=\sum_{i=1}^{n}\sum_{j=1}^{k}z_{ij}^{\left(r\right)}\left[v_{ij}^{\left(r\right)}\ln \alpha_{\bX|j}+\left(1-v_{ij}^{\left(r\right)}\right)\ln \left(1-\alpha_{\bX|j}\right)\right],\\
&Q_3\left(\bmu_{\bX},\bSigma_{\bX}|\bvartheta^{\left(r\right)}\right)=-\frac{1}{2}\sum_{i=1}^n\sum_{j=1}^k\Biggl\{z_{ij}^{\left(r\right)}\ln\left|\bSigma_{\bX|j}\right|+z_{ij}^{\left(r\right)}\left(v_{ij}^{\left(r\right)}+\frac{1-v_{ij}^{\left(r\right)}}{\eta_{\bX|j}^{\left(r\right)}}\right)\delta\left(\bx_i,\bmu_{\bX|j};\bSigma_{\bX|j}\right)\Biggr\},\\
&Q_4\left(\balpha_{\bY}|\bvartheta^{\left(r\right)}\right)=\sum_{i=1}^{n}\sum_{j=1}^{k}z_{ij}^{\left(r\right)}\left[u_{ij}^{\left(r\right)}\ln \alpha_{\bY|j}+\left(1-u_{ij}^{\left(r\right)}\right)\ln \left(1-\alpha_{\bY|j}\right)\right],\\
&Q_5\left(\bbeta,\bSigma_{\bY}|\bvartheta^{\left(r\right)}\right)=-\frac{1}{2}\sum_{i=1}^n\sum_{j=1}^k\Biggl\{z_{ij}^{\left(r\right)}\ln\left|\bSigma_{\bY|j}\right|+z_{ij}^{\left(r\right)}\left(u_{ij}^{\left(r\right)}+\frac{1-u_{ij}^{\left(r\right)}}{\eta_{\bY|j}^{\left(r\right)}}\right)\delta\left(\bx_i,\bmu_{\bY}\left(\bx_i;\bbeta_j\right);\bSigma_{\bY|j}\right)\Biggr\}.
\end{align*}
Terms which are independent by the parameters of interest have been removed from $Q_3$ and $Q_5$. 
As the five terms on the right-hand side of \eqref{eq:global Q} have zero cross-derivatives, they can be maximized separately.

\subsection{Update of $\bpi$}
\label{subapp:bpi}

The maximum of $Q_1\left(\bpi|\bvartheta^{\left(r\right)}\right)$ with respect to $\bpi$, subject to the constraints on those parameters, is obtained by maximizing the augmented function
\begin{equation}
\sum_{i=1}^n\sum_{j=1}^kz_{ij}^{\left(r\right)}\ln\pi_j-\lambda\left(\sum_{j=1}^k\pi_j-1\right),
\label{eq:Lagrange}
\end{equation}
where $\lambda$ is a Lagrangian multiplier.
Setting the derivative of equation \eqref{eq:Lagrange} with respect to $\pi_j$ equal to zero and solving for $\pi_j$ yields
\begin{equation*}
\pi_j^{\left(r+1\right)}=\displaystyle\displaystyle\sum_{i=1}^nz_{ij}^{\left(r\right)}\Big/n.
\end{equation*} 

\subsection{Update of $\balpha_{\bX}$}
\label{subapp:balphaX}

The updates for $\balpha_{\bX}$ can be obtained through the first partial derivatives
\begin{align}
	\frac{\partial Q_2\left(\balpha_{\bX}|\bvartheta^{\left(r\right)}\right)}{\partial \alpha_{\bX|j}} & = \frac{1}{\alpha_{\bX|j}}\sum_{i=1}^n z_{ij}^{\left(r\right)}v_{ij}^{\left(r\right)}-\frac{1}{1-\alpha_{\bX|j}}\sum_{i=1}^n z_{ij}^{\left(r\right)}\left(1-v_{ij}^{\left(r\right)}\right)
	= \frac{\displaystyle\sum_{i=1}^n z_{ij}^{\left(r\right)}v_{ij}^{\left(r\right)}-\alpha_{\bX|j}\sum_{i=1}^n z_{ij}^{\left(r\right)}}{\alpha_{\bX|j}\left(1-\alpha_{\bX|j}\right)},\qquad j=1,\ldots,k.
	\label{eq:alphaX derivative}
\end{align}
Equating \eqref{eq:alphaX derivative} to zero yields
\begin{displaymath}
	\alpha_{\bX|j}^{\left(r+1\right)}=\displaystyle\sum_{i=1}^n z_{ij}^{\left(r\right)}v_{ij}^{\left(r\right)}\Bigg/\displaystyle\sum_{i=1}^n z_{ij}^{\left(r\right)},\qquad j=1,\ldots,k.
\end{displaymath}

\subsection{Update of $\bmu_{\bX}$ and $\bSigma_{\bX}$}
\label{subapp:bmuX and bSigmaX}

The updates for $\bmu_{\bX}$ can be obtained through the first partial derivatives
\begin{align}
	\frac{\partial Q_3\left(\bmu_{\bX},\bSigma_{\bX}|\bvartheta^{\left(r\right)}\right)}{\partial \bmu_{\bX|j}} & = \frac{\partial \left[-\displaystyle\frac{1}{2}\displaystyle\sum_{i=1}^nz_{ij}^{\left(r\right)}\left(v_{ij}^{\left(r\right)}+\frac{1-v_{ij}^{\left(r\right)}}{\eta_{\bX|j}^{\left(r\right)}}\right)\left(\bx_i-\bmu_{\bX|j}\right)'\bSigma_{\bX|j}^{-1}\left(\bx_i-\bmu_{\bX|j}\right)\right]}{\partial \bmu_{\bX|j}}\nonumber\\
	& = \sum_{i=1}^nz_{ij}^{\left(r\right)}\left(v_{ij}^{\left(r\right)}+\frac{1-v_{ij}^{\left(r\right)}}{\eta_{\bX|j}^{\left(r\right)}}\right)\bSigma_{\bX|j}^{-1}\left(\bx_i-\bmu_{\bX|j}\right), \qquad j=1,\ldots,k.
	\label{eq:muX derivative}
\end{align}
Equating \eqref{eq:muX derivative} to the null vector yields
\begin{displaymath}
\bmu_{\bX|j}^{\left(r+1\right)}=\sum_{i=1}^n\frac{z_{ij}^{\left(r\right)}\displaystyle\left(v_{ij}^{\left(r\right)}+\frac{1-v_{ij}^{\left(r\right)}}{\eta_{\bX|j}^{\left(r\right)}}\right)}{\displaystyle\sum_{q=1}^n{z}_{qj}^{\left(r\right)}\left(v_{qj}^{\left(r\right)}+\frac{1-v_{qj}^{\left(r\right)}}{\eta_{\bX|j}^{\left(r\right)}}\right)}\bx_i, \qquad j=1,\ldots,k.	
\end{displaymath}

The updates for $\bSigma_{\bX}$, using results from matrix derivatives and trace operator \citep[see, e.g.,][]{Lutk:Hand:1996}, can be obtained through the first partial derivatives
\begin{align}
	\frac{\partial Q_3\left(\bmu_{\bX},\bSigma_{\bX}|\bvartheta^{\left(r\right)}\right)}{\partial \bSigma_{\bX|j}^{-1}} 
	& = \frac{\partial \left\{\displaystyle\frac{1}{2}\displaystyle\sum_{i=1}^n z_{ij}^{\left(r\right)} \left\{\ln\left|\bSigma_{\bX|j}^{-1}\right| - \left(v_{ij}^{\left(r\right)}+\displaystyle\frac{1-v_{ij}^{\left(r\right)}}{\eta_{\bX|j}^{\left(r\right)}}\right)\tr\left[\left(\bx_i-\bmu_{\bX|j}^{\left(r+1\right)}\right)'\bSigma_{\bX|j}^{-1}\left(\bx_i-\bmu_{\bX|j}^{\left(r+1\right)}\right)\right]\right\}\right\}}{\partial \bSigma_{\bX|j}^{-1}}\nonumber\\
	& = \frac{\partial \left\{\displaystyle\frac{1}{2}\displaystyle\sum_{i=1}^n z_{ij}^{\left(r\right)} \left\{\ln\left|\bSigma_{\bX|j}^{-1}\right| - \left(v_{ij}^{\left(r\right)}+\displaystyle\frac{1-v_{ij}^{\left(r\right)}}{\eta_{\bX|j}^{\left(r\right)}}\right)\tr\left[\bSigma_{\bX|j}^{-1}\left(\bx_i-\bmu_{\bX|j}^{\left(r+1\right)}\right)\left(\bx_i-\bmu_{\bX|j}^{\left(r+1\right)}\right)'\right]\right\}\right\}}{\partial \bSigma_{\bX|j}^{-1}}\nonumber\\
	&=\displaystyle\frac{1}{2}\displaystyle\sum_{i=1}^n z_{ij}^{\left(r\right)} \left[\bSigma_{\bX|j}+ \left(v_{ij}^{\left(r\right)}+\displaystyle\frac{1-v_{ij}^{\left(r\right)}}{\eta_{\bX|j}^{\left(r\right)}}\right)\left(\bx_i-\bmu_{\bX|j}^{\left(r+1\right)}\right)\left(\bx_i-\bmu_{\bX|j}^{\left(r+1\right)}\right)'\right], \qquad j=1,\ldots,k.
	\label{eq:SigmaX derivative}
\end{align}
Equating \eqref{eq:SigmaX derivative} to the null matrix yields
\begin{displaymath}
	\bSigma_{\bX|j}^{\left(r+1\right)} =\frac{1}{\displaystyle\sum_{i=1}^n z_{ij}^{\left(r\right)}}\sum_{i=1}^nz_{ij}^{\left(r\right)}\left(v_{ij}^{\left(r\right)}+\frac{1-v_{ij}^{\left(r\right)}}{\eta_{\bX|j}^{\left(r\right)}}\right)\left(\bx_i-\displaystyle\bmu_{\bX|j}^{\left(r+1\right)}\right)\left(\bx_i-\displaystyle\bmu_{\bX|j}^{\left(r+1\right)}\right)', \qquad j=1,\ldots,k.
\end{displaymath}

\subsection{Update of $\balpha_{\bY}$}
\label{subapp:balphaY}

Analogously to \ref{subapp:balphaX}, the updates for $\balpha_{\bY}$ can be obtained through the first partial derivatives
\begin{equation}
	\frac{\partial Q_4\left(\balpha_{\bY}|\bvartheta^{\left(r\right)}\right)}{\partial \alpha_{\bY|j}} = \frac{\displaystyle\sum_{i=1}^n z_{ij}^{\left(r\right)}u_{ij}^{\left(r\right)}-\alpha_{\bY|j}\sum_{i=1}^n z_{ij}^{\left(r\right)}}{\alpha_{\bY|j}\left(1-\alpha_{\bY|j}\right)}, \qquad j=1,\ldots,k.
	\label{eq:alphaY derivative}
\end{equation}
Equating \eqref{eq:alphaY derivative} to zero yields
\begin{displaymath}
	\alpha_{\bY|j}^{\left(r+1\right)}=\displaystyle\sum_{i=1}^n z_{ij}^{\left(r\right)}u_{ij}^{\left(r\right)}\Bigg/\displaystyle\sum_{i=1}^n z_{ij}^{\left(r\right)}, \qquad j=1,\ldots,k.
\end{displaymath}

\subsection{Update of $\bbeta$ and $\bSigma_{\bY}$}
\label{subapp:bbeta and bSigmaY}

Using properties of trace and transpose, the updates for $\bbeta$ can be obtained through the first partial derivatives
\begin{align}
	\frac{\partial Q_5\left(\bbeta,\bSigma_{\bY}|\bvartheta^{\left(r\right)}\right)}{\partial \bbeta_j'} & = \frac{\partial \left\{-\displaystyle\frac{1}{2}\displaystyle\sum_{i=1}^nz_{ij}^{\left(r\right)}\left(u_{ij}^{\left(r\right)}+\frac{1-u_{ij}^{\left(r\right)}}{\eta_{\bY|j}^{\left(r\right)}}\right)\left[\by_i-\bmu_{\bY}\left(\bx_i;\bbeta_j\right)\right]'\bSigma_{\bX|j}^{-1}\left[\by_i-\bmu_{\bY}\left(\bx_i;\bbeta_j\right)\right]\right\}}{\partial \bbeta_j'}\nonumber\\
	&=  \frac{\partial \left[-\displaystyle\frac{1}{2}\displaystyle\sum_{i=1}^nz_{ij}^{\left(r\right)}\left(u_{ij}^{\left(r\right)}+\frac{1-u_{ij}^{\left(r\right)}}{\eta_{\bY|j}^{\left(r\right)}}\right)\left(-\by_i'\bSigma_{\bY|j}^{-1}\bbeta_j'\bx_i^*-\bx_i^{*'}\bbeta_j\bSigma_{\bY|j}^{-1}\by_i+\bx_i^{*'}\bbeta_j\bSigma_{\bY|j}^{-1}\bbeta_j'\bx_i^*\right)\right]}{\partial \bbeta_j'}	\nonumber\\
	&=  \frac{\partial \left\{\displaystyle\frac{1}{2}\displaystyle\sum_{i=1}^nz_{ij}^{\left(r\right)}\left(u_{ij}^{\left(r\right)}+\frac{1-u_{ij}^{\left(r\right)}}{\eta_{\bY|j}^{\left(r\right)}}\right)\left[\tr\left(\bbeta_j'\bx_i^*\by_i'\bSigma_{\bY|j}^{-1}\right)+\tr\left(\left(\bSigma_{\bY|j}^{-1}\by_i\bx_i^{*'}\right)'\bbeta_j'\right)-\tr\left(\bbeta_j'\bx_i^*\bx_i^{*'}\bbeta_j\bSigma_{\bY|j}^{-1}\right)\right]\right\}}{\partial \bbeta_j'}\nonumber\\
	&= \displaystyle\frac{1}{2}\displaystyle\sum_{i=1}^nz_{ij}^{\left(r\right)}\left(u_{ij}^{\left(r\right)}+\frac{1-u_{ij}^{\left(r\right)}}{\eta_{\bY|j}^{\left(r\right)}}\right)\left(2\bSigma_{\bY|j}^{-1}\by_i\bx_i^{*'}-2\bSigma_{\bY|j}^{-1}\bbeta_j'\bx_i^*\bx_i^{*'}\right), \qquad j=1,\ldots,k.
	\label{eq:bbeta derivative}
\end{align}
Equating \eqref{eq:bbeta derivative} to the null matrix yields
\begin{displaymath}
\bbeta_j^{\left(r+1\right)} = \left[\sum_{i=1}^n z_{ij}^{\left(r\right)}\left(u_{ij}^{\left(r\right)}+\frac{1-u_{ij}^{\left(r\right)}}{\eta_{\bY|j}^{\left(r\right)}}\right)\bx_i^*\bx_i^{*'}\right]^{-1}\left[\sum_{i=1}^n z_{ij}^{\left(r\right)}\left(u_{ij}^{\left(r\right)}+\frac{1-u_{ij}^{\left(r\right)}}{\eta_{\bY|j}^{\left(r\right)}}\right)\bx_i^*\by_i\right], \qquad j=1,\ldots,k.	
\end{displaymath}

Finally, the updates for $\bSigma_{\bY}$ can be obtained, analogously to the updates for $\bSigma_{\bX}$ given in \ref{subapp:bmuX and bSigmaX}, through the first partial derivatives
\begin{equation}
	\frac{\partial Q_5\left(\bbeta,\bSigma_{\bY}|\bvartheta^{\left(r\right)}\right)}{\partial \bSigma_{\bY|j}^{-1}} 
	=\displaystyle\frac{1}{2}\displaystyle\sum_{i=1}^n z_{ij}^{\left(r\right)} \left\{\bSigma_{\bY|j}+ \left(u_{ij}^{\left(r\right)}+\displaystyle\frac{1-u_{ij}^{\left(r\right)}}{\eta_{\bY|j}^{\left(r\right)}}\right)\left[\by_i-\bmu_{\bY}\left(\bx_i;\bbeta_j^{\left(r+1\right)}\right)\right]\left[\by_i-\bmu_{\bY}\left(\bx_i;\bbeta_j^{\left(r+1\right)}\right)\right]'\right\}, \qquad j=1,\ldots,k.
	\label{eq:SigmaY derivative}
\end{equation}
Equating \eqref{eq:SigmaY derivative} to the null matrix yields
\begin{displaymath}
	\bSigma_{\bY|j}^{\left(r+1\right)} =\frac{1}{n_j^{\left(r\right)}}\sum_{i=1}^nz_{ij}^{\left(r\right)}\left(u_{ij}^{\left(r\right)}+\frac{1-u_{ij}^{\left(r\right)}}{\eta_{\bY|j}^{\left(r\right)}}\right)\left[\by_i-\displaystyle\bmu_{\bY}\left(\bx_i;\bbeta_j^{\left(r+1\right)}\right)\right]\left[\by_i-\displaystyle\bmu_{\bY}\left(\bx_i;\bbeta_j^{\left(r+1\right)}\right)\right]', \qquad j=1,\ldots,k.
\end{displaymath}

\section*{References}

\bibliographystyle{elsarticle-harv} 


\end{document}